\documentclass[onecolumn,10pt]{IEEEtran}
\usepackage{amsmath}
\usepackage{indentfirst,mathrsfs}
\usepackage{amsfonts,caption2}
\usepackage{amsmath,amsthm,amssymb}
\usepackage{color}
\usepackage{cite}

\usepackage{tabularx}
\usepackage{multirow}
\usepackage{booktabs}
\usepackage{tabularx}
\usepackage{array}
\usepackage{float}
\usepackage{ragged2e}
\usepackage{placeins}
\usepackage{multirow}
\usepackage{booktabs}
\usepackage{makecell}
\usepackage{titlesec}
\newcolumntype{L}[1]
{>{\RaggedRight\arraybackslash}p{#1}}

\newtheorem{Theorem}{Theorem}
\newtheorem{Corollary}{Corollary}
\newtheorem{Definition}{Definition}
\newtheorem{Example}{Example}
\newtheorem{Remark}{Remark}
\newtheorem{Proposition}{Proposition}

\newtheorem{Lemma}{Lemma}
\newcommand{\F}{\mathbb{F}}

\allowdisplaybreaks[4]

\begin{document}
\title{Explicit Constructions of Maximum-Cardinality Families of Disjoint Spectra Plateaued Functions Without Nonzero Linear Structures }
\author{
Chen Wang, 
Xiaoyan Zhang, 
Chunming Tang, 
and Zhengchun Zhou,
\thanks{
The authors are with the School of Information Science and Technology, and the School of Mathematics at Southwest Jiaotong University, Chengdu 611756, China (E-mail: 
wangchen1131707473@163.com,
zhangxiaoyan@my.swjtu.edu.cn,
tangchunmingmath@163.com,
zzc@swjtu.edu.cn)
}
}

\maketitle

\begin{abstract}
Families of disjoint spectra plateaued Boolean functions are useful in
secondary constructions of cryptographic Boolean functions. Of
particular interest are maximum-cardinality families of such functions
without nonzero linear structures. To our knowledge, the only previous general construction achieving both properties is spectral.
In this paper, we present two new explicit algebraic constructions
within a unified framework, one based on linear functions and the other
on partially linear functions with bent components. Let \(p\geq 2\)
and \(q\geq 0\) satisfy \(q<2^p-p-1\), and set \(m=p+q\). Both
constructions yield maximum-cardinality families of \(2^{q+1}\)
disjoint spectra \((q+1)\)-plateaued Boolean functions without nonzero
linear structures. Every member has an explicit generalized
Maiorana--McFarland representation.
The first construction produces functions in \(m+p+1\) variables and
realizes any prescribed common algebraic degree
\(3\leq d\leq p+1\), provided that
\(q<\sum_{i=2}^{d-1}\binom{p}{i}\); its maximum attainable degree
\(p+1\) is optimal. The second construction produces functions in
\(n+p+1\) variables, where \(n>m\) and \(n-m\) is even, and realizes
any prescribed common algebraic degree
\(3\leq d\leq p+(n-m)/2\), provided that
\(q<\sum_{i=2}^{\min\{d-1,p\}}\binom{p}{i}\); its maximum attainable
degree \(p+(n-m)/2\) is next-to-optimal.
As an application, given a binary \([m,q,\delta]\) code with
\(\delta\geq 2t+1\), we use it to specify the first construction
and apply a common invertible linear transformation to obtain
a family of \(2^{q+1}-2\) \(t\)-resilient disjoint spectra
\((q+1)\)-plateaued functions without nonzero linear structures.
Based on an existing framework, we further obtain strictly almost optimal (SAO)
\(t\)-resilient Boolean functions without nonzero linear
structures and with optimal algebraic degree.
\end{abstract}


\section{Introduction}
Plateaued Boolean functions were introduced by Zheng and Zhang \cite{plateauedchushi} as a generalization of bent functions and have since received considerable attention because of their spectral regularity and applications in cryptography\cite{cry1,crp1,crp2}, coding theory\cite{code1,code2,code3,code4}, and sequence design \cite{sequence1,sequence2,sequence3}. Their Walsh spectra take values in $\{0,\pm 2^{(n+s)/2}\}$ for some admissible integer $s$, which makes them a natural class for studying the interaction among spectral support, algebraic degree, and linear structures.

Over the past two decades, numerous constructions of plateaued functions have been developed \cite{carletbook,framework,plateaued1,plateaued2,plateaued3,plateaued4,plateaued5,cry1,cry2,cry3}. Among them, the generalized Maiorana--McFarland $(\mathrm{GMM})$ construction, obtained by modifying the classical Maiorana--McFarland construction of bent functions, produces plateaued functions under suitable conditions \cite{carletbook,plateaued2}. Another important approach is the indirect-sum construction, which combines constituent Boolean functions satisfying appropriate spectral compatibility conditions \cite{carletbook}.

Disjoint spectra plateaued functions constitute an important class of plateaued Boolean functions. In recent years, such spectral-disjointness structures have been used in several secondary constructions. For example, maximum-cardinality families of $s$-plateaued functions can be concatenated to produce bent functions \cite{bent1,bent2}. Totally disjoint spectra plateaued functions have also been used to construct five-valued-spectrum Boolean functions \cite{5}. The formal study of families of disjoint spectra plateaued
functions dates back to 2009, when Zhang and Xiao introduced
the notion and developed a general framework for constructing
almost optimal resilient Boolean functions from large families
of such functions.
 Since their construction relied mainly
on partially linear functions, which possess nonzero linear
structures, they asked whether large such families could be constructed
with members inequivalent to partially linear functions.

In 2014, Zhang, Carlet, Hu, and Cao modified a generalized
indirect-sum construction to obtain quadruples of disjoint spectra plateaued functions
 and gave sufficient conditions
for all four functions to have no nonzero linear structures
\cite{2014disjoint}. This provided an important first step toward
addressing the problem of Zhang and Xiao, although the family size
was limited to four. Subsequently, Zhang, Wei, Pasalic, and Xia
developed a recursive method that progressively reduces the linear
kernel and uses basis conditions realized through primitive
polynomials and companion matrices to eliminate nonzero linear
structures \cite{disjoint1}. Their method produced large families
of disjoint spectra plateaued functions that are inequivalent to
partially linear functions, thereby resolving the 2009 problem.
Nevertheless, the resulting families did not attain the maximum
possible cardinality, and the construction required repeated
iterations together with nontrivial ordering and basis conditions.

In 2019, Hodžić, Pasalic, Wei, and Zhang developed a spectral
method that specifies the Walsh support together with the signs of
its nonzero coefficients, rather than starting from an algebraic
normal form \cite{disjoint2}. Their method yields a maximum-cardinality family of disjoint spectra
plateaued functions without nonzero linear structures. They also distinguished trivial plateaued
functions, whose affine Walsh supports make them partially bent and
force nonzero linear structures, from nontrivial plateaued functions,
whose Walsh supports are non-affine and which are EA-inequivalent to
the trivial class. In particular, a nontrivial plateaued function is
guaranteed to have no nonzero linear structures whenever its Walsh support affinely spans the ambient space. Thus, the 2019
work resolved the maximum-cardinality problem at the spectral level.
However, in the spectral construction, explicit algebraic normal forms are not generally provided, and the algebraic degrees of the family members remain undetermined; the degrees of the family members in the 2018 construction \cite{disjoint1} were likewise not determined. 

Very recently, Pasalic, Kudin, Hodžić, and Khan gave an explicit algebraic construction of
individual plateaued functions in $\mathcal{GMM}$ that
simultaneously attain the maximum possible algebraic degree and admit
no nonzero linear structures~\cite{GMMNLSO}. More precisely, they
considered functions of the form
$f(x,y)=x\cdot\varphi(y)\oplus h(y)$,
where suitable algebraic conditions on the injective mapping
$\varphi$ ensure plateauedness and maximum algebraic degree, while
the nonconstancy of $a\cdot\varphi(y)$ for every nonzero $a$ excludes
nonzero linear structures. Their construction therefore provides
explicit individual plateaued functions belonging to $\mathcal{GMM}$ with both properties. It does not, however, organize such functions into a
maximum-cardinality family of disjoint spectra plateaued functions. Thus, to the best of our knowledge, no explicit algebraic construction
is known that produces a maximum-cardinality family of disjoint spectra
plateaued functions whose members admit no nonzero linear structures
and have a prescribed common algebraic degree.

Building on these developments, we introduce two new explicit
algebraic constructions within a unified framework, based
respectively on linear functions and on partially linear functions
with bent components. Both constructions yield maximum-cardinality
families of \(2^{q+1}\) disjoint spectra \((q+1)\)-plateaued Boolean functions without nonzero linear structures. The first construction produces
functions in \(m+p+1\) variables and realizes every prescribed common
algebraic degree \(3\leq d\leq p+1\), whenever
$q<\sum_{i=2}^{d-1}\binom{p}{i}$;
in particular, its maximum attainable degree \(p+1\) is optimal.
The second construction produces functions in \(n+p+1\) variables,
where \(n>m\) and \(n-m\) is even, and realizes every prescribed
common algebraic degree
$3\leq d\leq p+(n-m)/2$,
whenever
$q<\sum_{i=2}^{\min\{d-1,p\}}\binom{p}{i}$;
its maximum attainable degree \(p+(n-m)/2\) is next-to-optimal.
In both constructions, every member has an explicit $\mathrm{GMM}$ representation.

Moreover, for either construction, if $q\neq0$ and
$\gamma_k\neq\gamma_{k'}$ for some $k\neq k'$, then the Walsh
supports of the resulting family are not all translations of a
single fixed support. In contrast, every family obtained from
\cite[Theorem~4.4]{disjoint2} has the common-translation property.
Since this property is preserved under a common EA transformation,
any family obtained from either of our constructions under the
preceding condition is not commonly EA-equivalent to any family
obtained from that spectral construction. We further show that the
Walsh-support self-sum cardinality is invariant under EA
transformations and use this invariant to give concrete examples
from both constructions that are not memberwise EA-equivalent to
any family obtained from the spectral construction.

As a further application of the first construction, suppose that
a binary \([m,q,\delta]\) linear code with \(\delta\ge 2t+1\) exists.
Using this code and a common invertible linear transformation,
we obtain a family of \(2^{q+1}-2\) \(t\)-resilient disjoint
spectra \((q+1)\)-plateaued functions without nonzero linear
structures. A general counting upper bound shows that no such family in
\(m+p+1\) variables can contain more than \(2^{q+1}-1\)
functions. Consequently, under the above code-existence
assumption, the maximum attainable cardinality is either
\(2^{q+1}-2\) or \(2^{q+1}-1\), so our family is either
optimal or falls short of the optimum by only one function.
If, in addition,
$\sum_{i=0}^{t}\binom{m+p+1}{i}\le 2^{q+1}-2$,
we then select
$\sum_{i=0}^{t}\binom{m+p+1}{i}$ members from this family,
index them by the vectors in $\mathbb F_2^{m+p+1}$ of
Hamming weight at most $t$, and use them as constituents
in the framework of \cite{2009disjoint}. This
produces two classes of SAO
\(t\)-resilient Boolean functions without nonzero linear
structures. The first class has the higher guaranteed
nonlinearity but does not in general attain the optimal
algebraic degree, whereas the second attains the optimal
algebraic degree at the cost of a lower guaranteed
nonlinearity. Explicit instances for \(t=1,2,\) and \(3\)
demonstrate the applicability of this approach to different
resiliency orders.

Table~\ref{tab:comparison} summarizes the main 
differences between the 
proposed constructions and 
the most closely related 
works.

The remainder of the paper is organized as follows.
Section~\ref{pre} reviews the necessary background.
Section~\ref{sec-0} introduces a unified framework for constructing
maximum-cardinality families of disjoint spectra plateaued functions. Section~\ref{sec-plateauedwithoutlinear} develops the first
construction based on linear functions and compares the resulting
families with the spectral construction.
Explicit coordinate forms are then used to determine the
attainable common algebraic degrees.
Section~\ref{sec-partially-bent} presents the second
construction, based on partially linear functions with bent
components. As in the preceding section, it compares the
resulting families with the spectral construction and uses
coordinate forms to determine their attainable common
algebraic degrees. As an application of the first construction,
Section~\ref{sec-SAO} obtains a large family of
\(t\)-resilient disjoint spectra plateaued functions without
nonzero linear structures and constructs SAO \(t\)-resilient
Boolean functions without nonzero linear structures and with
optimal algebraic degree. Finally, Section~\ref{con} concludes the paper.
\begin{table}[H]\footnotesize
\centering
\caption{Comparison of the algebraic and structural properties of closely related constructions.}
\label{tab:comparison}
\setlength{\tabcolsep}{3pt}
\renewcommand{\arraystretch}{1.15}
\renewcommand{\tabularxcolumn}[1]{m{#1}}

\begin{tabularx}{\textwidth}{
    @{}
    >{\raggedright\arraybackslash}m{0.22\textwidth}
    *{5}{>{\centering\arraybackslash}X}
    @{}
}
\toprule
\multicolumn{1}{c}{\multirow{2}{*}{Property}}
& \multicolumn{3}{c}{Related works}
& \multicolumn{2}{c}{This work} \\
\cmidrule(lr){2-4}
\cmidrule(lr){5-6}
&
2018 [1]
&
2019 [2]
&
2026 [20]
&
$\Omega_R^{\mathcal{F}_0}$
&
$\Omega_R^{\mathcal{F}_1}$ \\
\midrule

Construction scope
& Family
& Family
& Individual
& Family
& Family \\

Family cardinality
& $<2^s$
& $2^s$
& N/A
& $2^{q+1}$
& $2^{q+1}$ \\

Maximum cardinality attained
& No
& Yes
& N/A
& Yes
& Yes \\

\addlinespace[2pt]

Construction method
& Recursive
& Spectral
& \shortstack{Explicit $\mathrm{GMM}$}
& \shortstack{Explicit algebraic}
& \shortstack{Explicit algebraic} \\

Family of disjoint spectra functions
& Yes
& Yes
& N/A
& Yes
& Yes \\

No nonzero linear structures
& Yes
& Yes
& Yes
& Yes
& Yes \\

\addlinespace[2pt]

Algebraic-degree control
& \shortstack{Not established}
& \shortstack{Not established}
& Optimal
& $3\leq d\leq p+1$
& \shortstack{$3\leq d\leq p+\tfrac{n-m}{2}$} \\

Relationship to $\mathcal{GMM}$
& \shortstack{Not established}
& \shortstack{Not established}
& Direct
& Direct
& Direct \\

\bottomrule
\end{tabularx}

\smallskip
\begin{minipage}{\textwidth}
\footnotesize
Note. For both constructions presented in this work,
$m=p+q$ and $s=q+1$.  The families
$\Omega_R^{\mathcal F_0}$ and $\Omega_R^{\mathcal F_1}$
are obtained from linear functions and partially linear
functions with bent components, respectively. Their maximum
attainable algebraic degrees are $p+1$ (optimal) and
$p+(n-m)/2$ (next-to-optimal), respectively.
\end{minipage}
\end{table}
\FloatBarrier

\section{Preliminaries}\label{pre}
An $n$-variable Boolean function is a mapping
$f:\mathbb F_2^n\rightarrow\mathbb F_2$.
The set of all $n$-variable Boolean functions is denoted by
$\mathcal B_n$.

For each integer $0\leq i\leq 2^n-1$, write
$i=\sum_{j=0}^{n-1}i_j2^j$ with $i_j\in\mathbb F_2$,
and associate $i$ with the vector
$\boldsymbol{i}=(i_0,\ldots,i_{n-1})\in\mathbb F_2^n$.
The truth table of $f$ is the binary vector
$\bigl(f(\boldsymbol{0}),f(\boldsymbol{1}),\ldots,
f(\boldsymbol{2^n-1})\bigr)$.

Every $f\in\mathcal B_n$ has a unique algebraic normal form (ANF)
\[
f(x)=
\bigoplus_{u\in\mathbb F_2^n}
a_u x^u,
\qquad
x^u=\prod_{j=0}^{n-1} x_j^{u_j},
\]
where $x=(x_0,\ldots,x_{n-1})$, and the coefficient $a_u\in\mathbb F_2$ is given by
\begin{equation}\label{Mob}
a_u=\bigoplus_{v\preceq u} f(v),
\end{equation}
with $v\preceq u$ denoting $v_j\leq u_j$ for all $0\leq j\leq n-1$.
For a nonzero Boolean function $f$, its algebraic degree is
$\deg(f)=
\max\bigl\{\operatorname{wt}(u)\,|\,a_u\neq 0\bigr\}$,
where $\operatorname{wt}(u)$ denotes the Hamming weight of $u$.
By convention, $\deg(0)=0$.

For \(f\in\mathcal B_n\), its Walsh transform at
\(\omega\in\mathbb F_2^n\) is defined by
\[
W_f(\omega)=
\sum_{x\in\mathbb F_2^n}
(-1)^{f(x)\oplus\omega\cdot x}.
\]
The Walsh support of \(f\) is
$\operatorname{supp}(W_f)=
\{\omega\in\mathbb F_2^n\mid W_f(\omega)\neq 0\}$.
In terms of Walsh spectra, the nonlinearity of $f$ is given by
\[
N_f
=2^{n-1}-\frac{1}{2}
\max_{\omega\in\mathbb F_2^n}|W_f(\omega)|.
\]
An \(n\)-variable Boolean function \(f\) is called strictly
almost optimal (SAO) if
\(N_f>2^{n-1}-2^{\lfloor n/2\rfloor}\).

The Walsh transform satisfies Parseval's identity
\(\sum_{\omega\in\mathbb F_2^n}W_f^2(\omega)=2^{2n}\).
Hence,
\(\max_{\omega\in\mathbb F_2^n}|W_f(\omega)|
\geq 2^{n/2}\).
For even \(n\), a Boolean function \(f\in\mathcal B_n\) is
called bent if
\(W_f(\omega)\in\{-2^{n/2},2^{n/2}\}\) for every
\(\omega\in\mathbb F_2^n\).
A bent function has nonlinearity
\(N_f=2^{n-1}-2^{n/2-1}\), the maximum possible for even \(n\).
For a bent function \(f\), its dual
\(\widetilde f\in\mathcal B_n\) is uniquely defined by
\(W_f(\omega)=2^{n/2}(-1)^{\widetilde f(\omega)}\),
where \(\omega\in\mathbb F_2^n\).
The dual \(\widetilde f\) is also bent.

Bent functions also admit an equivalent characterization in terms
of their derivatives. Specifically, \(f\) is bent if and only if
\(D_\alpha f(x)=f(x)\oplus f(x\oplus\alpha)\) is balanced for
every \(\alpha\in\mathbb F_2^n\setminus\{\boldsymbol{0}_n\}\)
\cite{carletbook}. Consequently, a bent function
admits no nonzero linear structure.
If \(n \geq 4\), the algebraic degree of a bent function is at most
\(n/2\), and every integer degree between \(2\) and \(n/2\) is
attainable. Every bent function in two variables has algebraic
degree \(2\).

More generally, \(f\) is called \(s\)-plateaued if
$W_f(\omega)\in\{0,\pm2^{(n+s)/2}\}$
for every \(\omega\in\mathbb F_2^n\), where
\(0\leq s\leq n\) and \(n+s\) is even. In this case,
$|\operatorname{supp}(W_f)|=2^{n-s}$ and
$\deg(f)\leq (n-s)/2+1$.
A plateaued function is called semi-bent when \(s=1\) for odd
\(n\) and \(s=2\) for even \(n\). An \(s\)-plateaued function
with \(s>0\) is said to have optimal algebraic degree if
$\deg(f)=(n-s)/2+1$.

We next recall the generalized Maiorana--McFarland
$(\mathrm{GMM})$ construction. 

\begin{Definition}
Let $r,t$ be positive integers. A Boolean function
$g\in\mathcal{B}_{r+t}$ of the form
\[
g(x,y)=x\cdot\pi(y)\oplus h(y),
\qquad
x\in\mathbb{F}_2^r,\quad y\in\mathbb{F}_2^t,
\]
where $\pi:\mathbb{F}_2^t\to\mathbb{F}_2^r$ is a mapping and
$h\in\mathcal{B}_t$, is called a generalized
Maiorana--McFarland $(\mathrm{GMM})$ function. The set of all such functions,
as $\pi$ and $h$ vary, is called the generalized
Maiorana--McFarland class and is denoted by
$\mathcal{GMM}$.
\end{Definition}

Two standard sufficient conditions for a $\mathrm{GMM}$ function
to be plateaued are as follows. If $\pi$ is injective, which
necessarily requires $t\leq r$, then $g$ is $(r-t)$-plateaued.
If $\pi$ is exactly two-to-one onto its image, which necessarily
requires $t\leq r+1$, then $g$ is $(r-t+2)$-plateaued
~\cite{carletbook,plateaued2}.

Since the constructions developed below depend on prescribed
relations among Walsh supports, we adopt the following terminology.
\begin{Definition}
Let $\mathcal F=\{f_i\mid i\in I\}\subseteq\mathcal B_n$
be a family of Boolean functions. We call $\mathcal F$ a
family of disjoint spectra functions if
$\operatorname{supp}(W_{f_i})\cap
\operatorname{supp}(W_{f_j})=\varnothing$
for all distinct $i,j\in I$. If every $f_i$ is plateaued,
we call $\mathcal F$ a
family of disjoint spectra plateaued functions.
\end{Definition}

\begin{Definition}
    Let $1\leq m<n$, and write
$x=(x',x'')\in\mathbb F_2^{\,n-m}\times\mathbb F_2^m$.
A Boolean function $f\in\mathcal B_n$ is called a
partially linear function if it can be written as
    $f(x',x'')=h(x')\oplus c\cdot x''$,
where $h\in\mathcal B_{n-m}$ and $c\in\mathbb F_2^m$.
\end{Definition}
\begin{Lemma}[{\!\!\cite[Lemma~3]{2009disjoint}}]
For each \(c\in\mathbb F_2^m\), choose
\(h_c\in\mathcal B_{n-m}\) and define
    $f_c(x',x'')=h_c(x')\oplus c\cdot x''$.
Then \(\{f_c\mid c\in\mathbb F_2^m\}\) is a family
of disjoint spectra functions.
\end{Lemma}

If $n-m$ is even and $h_c$ is bent for every
$c\in\mathbb F_2^m$, then each $f_c$ is an $n$-variable
$m$-plateaued function. Hence,
$\{f_c\,|\,c\in\mathbb F_2^m\}$ is a maximum-cardinality family
of $2^m$ $n$-variable disjoint spectra $m$-plateaued functions.

\begin{Definition}
Let \(0\leq t<n\). An \(n\)-variable Boolean function
\(f\in\mathcal B_n\) is \(t\)-resilient if and only if
\(W_f(\omega)=0\) for every \(\omega\in\mathbb F_2^n\)
with \(\operatorname{wt}(\omega)\leq t\).
\end{Definition}

A \(t\)-resilient function
\(f\in\mathcal B_n\), where \(0\leq t\leq n-2\), satisfies
\(\deg(f)\leq n-t-1\). It is said to have
optimal algebraic degree if
\(\deg(f)=n-t-1\).

To determine whether the constructed functions admit nonzero
linear structures, we recall derivatives, linear structures and their behavior
under EA-equivalence.

Let $f\in\mathcal{B}_n$. For $\gamma\in\mathbb{F}_2^n$, the derivative
of $f$ in the direction $\gamma$ is the Boolean function
$D_\gamma f\in\mathcal{B}_n$ defined by
$D_\gamma f(x)
=
f(x)\oplus f(x\oplus\gamma),
$ for $x\in\F_2^n$.
A vector \(\gamma\in\mathbb F_2^n\) is called a linear
structure of \(f\) if \(D_\gamma f\) is constant.

\begin{Definition}
Two Boolean functions $f,g\in\mathcal{B}_n$ are said to be
 EA-equivalent if there exist
$A\in GL(n,2)$, $a,b\in\mathbb{F}_2^n$, and $c\in\mathbb{F}_2$
such that
$g(x)=f(Ax\oplus a)\oplus b\cdot x\oplus c$
for every $x\in\mathbb{F}_2^n$.
In particular, when $a=b=0$ and $c=0$, EA-equivalence reduces
to linear equivalence. Thus, linear equivalence is a special
case of EA-equivalence.
\end{Definition}

\begin{Proposition}\label{Prop-linear-equivalence-invariants}
Let $f,g\in\mathcal{B}_n$ satisfy
$g(x)=f(Ax\oplus a)\oplus b\cdot x\oplus c$
for some $A\in GL(n,2)$, $a,b\in\mathbb{F}_2^n$, and
$c\in\mathbb{F}_2$. Then the following statements hold.
\begin{enumerate}
\item For every $\omega\in\mathbb{F}_2^n$,
$W_g(\omega)
=
(-1)^{c\oplus(\omega\oplus b)\cdot A^{-1}a}
W_f\bigl(A^{-T}(\omega\oplus b)\bigr)$.
Consequently,
$\operatorname{supp}(W_g)
=
A^T\operatorname{supp}(W_f)\oplus b$.

\item If $\deg(f)\geq 2$, then
$\deg(g)=\deg(f)$.

\item For every $\alpha\in\mathbb{F}_2^n$,
$D_\alpha g(x)
=
D_{A\alpha}f(Ax\oplus a)\oplus b\cdot\alpha$.
\end{enumerate}
\end{Proposition}

\begin{proof}
For every $\omega\in\mathbb{F}_2^n$, under the change of variables $z=Ax\oplus a$, we have
\[
\begin{aligned}
W_g(\omega)
&=
\sum_{x\in\mathbb{F}_2^n}
(-1)^{f(Ax\oplus a)\oplus b\cdot x\oplus c
      \oplus\omega\cdot x}\\
&=
(-1)^c
\sum_{z\in\mathbb{F}_2^n}
(-1)^{f(z)\oplus(\omega\oplus b)\cdot
      A^{-1}(z\oplus a)}\\
&=
(-1)^{c\oplus(\omega\oplus b)\cdot A^{-1}a}
\sum_{z\in\mathbb{F}_2^n}
(-1)^{f(z)\oplus
      \left(A^{-T}(\omega\oplus b)\right)\cdot z}\\
&=
(-1)^{c\oplus(\omega\oplus b)\cdot A^{-1}a}
W_f\bigl(A^{-T}(\omega\oplus b)\bigr).
\end{aligned}
\]
Since $A\in GL(n,2)$,
it follows that
$\operatorname{supp}(W_g)
=
A^T\operatorname{supp}(W_f)\oplus b$.

Since invertible affine substitutions preserve algebraic degree and
$b\cdot x\oplus c$ is affine, we have $\deg(g)=\deg(f)$ whenever
$\deg(f)\geq 2$.

Finally, for every $\alpha,x\in\mathbb{F}_2^n$,\begin{equation*}
    D_\alpha g(x)
=
g(x)\oplus g(x\oplus\alpha)
=
f(Ax\oplus a)
\oplus f(Ax\oplus A\alpha\oplus a)
\oplus b\cdot\alpha=
D_{A\alpha}f(Ax\oplus a)\oplus b\cdot\alpha.
\end{equation*}
Therefore, $D_\alpha g$ is constant if and only if
$D_{A\alpha}f$ is constant.
\end{proof}

Proposition~\ref{Prop-linear-equivalence-invariants} implies that
EA-equivalence preserves plateauedness, the absence of nonzero linear
structures, and the algebraic degree of functions of degree at least
two, while a common EA transformation also preserves pairwise
disjointness of Walsh supports. These conclusions apply in particular
to linear equivalence.

Next, we introduce two notions of EA-equivalence for families of disjoint spectra plateaued Boolean functions.
\begin{Definition}
Let \(\mathcal{F}=\{f_i\mid i\in I\}\) and
\(\mathcal{G}=\{g_i\mid i\in I\}\) be two indexed families
of disjoint spectra plateaued Boolean functions on
\(\mathbb F_2^n\).
The families $\mathcal{F}$ and $\mathcal{G}$ are said to be
memberwise EA-equivalent if there exist a permutation $\rho$ of
$I$ and, for every $i\in I$, a matrix
$A_i\in\operatorname{GL}(n,2)$, vectors
$a_i,b_i\in\mathbb F_2^n$, and $c_i\in\mathbb F_2$ such that
\[
g_i(x)
=
f_{\rho(i)}(A_i x\oplus a_i)
\oplus b_i\cdot x\oplus c_i .
\]
If the same parameters can be used for all members, that is,
$A_i=A$, $a_i=a$, $b_i=b$, and $c_i=c$ for every $i\in I$, then
$\mathcal{F}$ and $\mathcal{G}$ are said to be commonly
EA-equivalent.
\end{Definition}
The following example exhibits two memberwise EA-equivalent
families of disjoint spectra plateaued Boolean functions.
\begin{Example}\label{Ex-memberwise-EA}
Let $x=(x_0,x_1)\in\mathbb F_2^2$ and
$y=(y_0,y_1)\in\mathbb F_2^2$. Define
$\mathcal{F}=\{f_0,f_1,f_2,f_3\}$ by
\[
\begin{aligned}
f_0(x,y)&=x_0x_1,\qquad\qquad\,
f_1(x,y)=x_0x_1\oplus y_0,\\
f_2(x,y)&=x_0x_1\oplus y_1,\qquad
f_3(x,y)=x_0x_1\oplus y_0\oplus y_1,
\end{aligned}
\]
and define $\mathcal{G}=\{g_0,g_1,g_2,g_3\}$ by
\[
\begin{aligned}
g_0(x,y)
&=x_0x_1\oplus x_0y_1\oplus x_1y_0\oplus y_0y_1
  \oplus x_1, \qquad
g_1(x,y)=x_0x_1\oplus x_0y_0\oplus x_1y_1\oplus y_0y_1
  \oplus x_0\oplus y_1\oplus1,\\
g_2(x,y)
&=x_0x_1\oplus x_0y_0\oplus x_1y_1\oplus y_0y_1
  \oplus x_1\oplus y_1, \qquad
  g_3(x,y)
=x_0x_1\oplus x_0y_1\oplus x_1y_0\oplus y_0y_1
  \oplus x_0.
\end{aligned}
\]
A direct calculation shows that both $\mathcal{F}$ and
$\mathcal{G}$ are families of four disjoint spectra $2$-plateaued Boolean functions.
To verify their memberwise EA-equivalence, let
\[
C_0=C_3=I_2,
\qquad
C_1=C_2=
\begin{pmatrix}
0&1\\
1&0
\end{pmatrix},
\qquad
A_i=
\begin{pmatrix}
I_2&C_i\\
0_2&I_2
\end{pmatrix}
\in\operatorname{GL}(4,2).
\]
Choose
\[
\begin{aligned}
a_0&=(1,0,0,0),&
a_1&=(0,1,0,0),&
a_2&=(0,0,1,0),&
a_3&=(0,0,0,1),\\
b_0&=(0,0,0,1),&
b_1&=(0,0,1,0),&
b_2&=(0,1,0,0),&
b_3&=(1,0,1,1),
\end{aligned}
\]
and $(c_0,c_1,c_2,c_3)=(0,1,0,1)$. Then, for every
$0\leq i\leq3$,
$g_i(x,y)
=
f_i\left(
A_i
\begin{pmatrix}
x\\
y
\end{pmatrix}
\oplus a_i
\right)
\oplus
b_i\cdot
\begin{pmatrix}
x\\
y
\end{pmatrix}
\oplus c_i$.
Hence, \(\mathcal{F}\) and \(\mathcal{G}\) are memberwise
EA-equivalent families of disjoint spectra plateaued functions.
\end{Example}

\section{The framework of constructing families of disjoint spectra plateaued functions}\label{sec-0}
In this section, we develop a unified framework based on an
auxiliary selector space and offset functions for constructing
maximum-cardinality families of disjoint spectra plateaued Boolean functions.

For each integer $k$ with $0\le k\le 2^q-1$, let
$S_k=\{s_{k0},\ldots,s_{k,2^p-1}\}$ be a $p$-dimensional
subspace of $\mathbb F_2^{p+1}$. Choose
$\gamma_k\in\mathbb F_2^{p+1}\setminus S_k$. Then
$\mathbb F_2^{p+1}=S_k\cup(\gamma_k\oplus S_k)$.
For every $0\le r\le 2^p-1$, set
$s'_{kr}=\gamma_k\oplus s_{kr}$.

For each $0\le k\le 2^q-1$, define the selector map
\begin{equation}\label{smap}
    \sigma_k:\mathbb F_2^{p+1}\rightarrow
\{0,1,\ldots,2^p-1\}
\end{equation}
by letting $\sigma_k(y)$ be the
unique integer $r$ such that $y\in\{s_{kr},s'_{kr}\}$.
This map is well defined because the pairs
$\{s_{kr},s'_{kr}\}$, where $0\le r\le 2^p-1$, form a
partition of $\mathbb F_2^{p+1}$.

Let $m=p+q$, and assume that $n\ge m$ and that $n+m$ is
even. For every $0\le k\le 2^q-1$ and $0\le r\le 2^p-1$,
let $f_{kr}$ be an $n$-variable $m$-plateaued Boolean
function. Suppose that these functions form a family of
disjoint spectra plateaued functions; that is,
$\operatorname{supp}(W_{f_{kr}})
\cap\operatorname{supp}(W_{f_{k'r'}})=\varnothing$
whenever $(k,r)\ne(k',r')$.

Let
\begin{equation}\label{Set-F}
\mathcal{F}
=
\left\{
f_{kr}
\ \middle|\
0\leq k\leq 2^q-1,\;
0\leq r\leq 2^p-1
\right\}
\end{equation}
denote the indexed family of constituent plateaued functions with
disjoint Walsh supports.

For offset parameters $a_{kr},a'_{kr},b_{kr},b'_{kr}\in\mathbb{F}_2$,
define
\begin{equation}\label{Set-R}
R
=
\left(
a_{kr},a'_{kr},b_{kr},b'_{kr}
\right)_{\substack{
0\leq k\leq 2^q-1\\
0\leq r\leq 2^p-1
}},
\end{equation}
where $a_{kr}\oplus a'_{kr}\oplus b_{kr}\oplus b'_{kr}
=
1$ for $0\leq k\leq 2^q-1$, $0\leq r\leq 2^p-1$.
For the offset parameters in $R$, define the offset
functions
\begin{equation*}
a_{k}(y)
=
\begin{cases}
a_{k,\sigma_k(y)}, & y\in S_k,\\
a'_{k,\sigma_k(y)}, & y\in\gamma_k\oplus S_k,
\end{cases}
\qquad
b_{k}(y)
=
\begin{cases}
b_{k,\sigma_k(y)}, & y\in S_k,\\
b'_{k,\sigma_k(y)}, & y\in\gamma_k\oplus S_k.
\end{cases}
\end{equation*}
We then define the Boolean functions
\begin{equation}\label{Eq-g_k}
\begin{aligned}
g_{k}(x,y)
&=
f_{k,\sigma_k(y)}(x)
\oplus a_{k}(y),\\
g'_{k}(x,y)
&=
f_{k,\sigma_k(y)}(x)
\oplus b_{k}(y),
\end{aligned}
\end{equation}
where
$(x,y)\in\mathbb{F}_2^n\times\mathbb{F}_2^{p+1}$.

Using the constituent functions in $\mathcal{F}$ and the offset
parameters in $R$, define the following family of Boolean functions
in $n+p+1$ variables:
\begin{equation}\label{Set-Omega}
\Omega_{R}^{\mathcal{F}}
=
\left\{
g_{k}\,|\, 0\leq k\leq 2^q-1
\right\}
\cup
\left\{
g'_{k}\,|\,0\leq k\leq 2^q-1
\right\},
\end{equation}
where $g_{k}$ and $g'_{k}$ are defined in~\eqref{Eq-g_k}.

\begin{Remark}
The subspaces \(S_k\) are not required to be pairwise distinct;
repetitions among them are allowed. For each \(k\), it suffices 
that \(S_k\) be a \(p\)-dimensional subspace of
\(\mathbb{F}_2^{p+1}\) and that \(\gamma_k\notin S_k\), so that
$\mathbb{F}_2^{p+1}
=
S_k\mathbin{\cup}(\gamma_k\oplus S_k)$. Explicit  constituent families $\mathcal{F}$
will be given in Sections \ref{sec-plateauedwithoutlinear} and \ref{sec-partially-bent}.
\end{Remark}
The following proposition establishes that the resulting
functions form a family of disjoint spectra plateaued functions.
\begin{Proposition}\label{The-disjoint}
   The family \(\Omega_{R}^{\mathcal F}\) defined in
\eqref{Set-Omega} is a maximum-cardinality family of \(2^{q+1}\) disjoint spectra
\((q+1)\)-plateaued functions in \(n+p+1\) variables.
\end{Proposition}
\begin{proof}
Since
$\mathbb{F}_2^{p+1}
=
S_k\mathbin{\cup}(\gamma_k\oplus S_k)$ and $S_k\mathbin{\cap}(\gamma_k\oplus S_k)=\emptyset$,
each \(y\in\mathbb{F}_2^{p+1}\) is exactly one of the
points \(s_{kr}\) and \(s'_{kr}\), where $0\leq r\leq2^p-1$.
Therefore,
 for arbitrary $(\omega',\omega'')\in\F_2^{n}\times\F_2^{p+1}$ and $g_{k}\in\Omega^{\mathcal{F}}_{R}$, we have
\begin{equation*}
     \begin{aligned}
         &W_{g_{k}}(\omega',\omega'')\\
=&
\sum_{(x,y)\in\F_2^{n}\times\mathbb{F}_2^{p+1}}
(-1)^{f_{k,\sigma_k(y)}(x)\oplus a_{k}(y)\oplus\omega'\cdot x\oplus\omega''\cdot y}\\
=& \sum_{x\in\mathbb{F}_2^n}(-1)^{\omega'\cdot x}
\sum_{r=0}^{2^p-1}
\left(
(-1)^{f_{kr}(x)\oplus a_{kr}\oplus \omega''\cdot s_{kr}}+(-1)^{f_{kr}(x)\oplus a'_{kr}\oplus \omega''\cdot s'_{kr}}
\right)\\
=& \sum_{r=0}^{2^p-1}
\left(
(-1)^{s_{kr}\cdot\omega''\oplus a_{kr}}
+(-1)^{s'_{kr}\cdot\omega''\oplus a'_{kr}}
\right)
\sum_{x\in\mathbb{F}_2^n}(-1)^{f_{kr}(x)\oplus\omega'\cdot x} \notag\\
=& \sum_{r=0}^{2^p-1}
\left(
(-1)^{s_{kr}\cdot\omega''\oplus a_{kr}}
+(-1)^{s'_{kr}\cdot\omega''\oplus a'_{kr}}
\right)
W_{f_{kr}}(\omega').
    \end{aligned}
\end{equation*}
   
For each fixed \(\omega'\), at most one term in the above sum is
nonzero, because the Walsh supports of the functions \(f_{kr}\)
are pairwise disjoint. Moreover,
$|
1+
(-1)^{
\gamma_k\cdot\omega''
\oplus a_{kr}\oplus a'_{kr}}
|
=2$
if and only if
$\gamma_k\cdot\omega''
\oplus a_{kr}\oplus a'_{kr}=0.$
Hence
\[
W_{g_{k}}(\omega',\omega'')
=
\begin{cases}
\pm2^{\frac{n+m+2}{2}},
&\omega'\in\displaystyle\operatorname{supp}(W_{f_{kr}}),\,\gamma_k\cdot\omega''\oplus a_{kr}\oplus a'_{kr}=0,\text{for some}  \,\,r,\,0\le r\le2^p-1,\\[6pt]
0,&\text{otherwise}.
\end{cases}
\]
Similarly, we have
\begin{equation*}
\begin{aligned}
   W_{g'_{k}}(\omega',\omega'')=\begin{cases}
\pm2^{\frac{n+m+2}{2}},
&\omega'\in\displaystyle\operatorname{supp}(W_{f_{kr}}),\,\gamma_k\cdot\omega''\oplus b_{kr}\oplus b'_{kr}=0,\text{for some}  \,\,r,\,0\le r\le2^p-1,\\[6pt]
0,&\text{otherwise}.
\end{cases}
\end{aligned}
\end{equation*}
Since $m=p+q$, each
$g_{k}$ and $g_{k}'$ is a $(q+1)$-plateaued function in $n+p+1$ variables, where  $0\leq k\leq 2^{q}-1$. 
Moreover, for each $k$, their Walsh supports are given by
\begin{equation}\label{Suppgg'}
    \begin{aligned}
&\mathrm{supp}(W_{g_{k}})=\bigcup_{r=0}^{2^p-1}\{(\omega',\omega'')\,|\,\omega'\in\operatorname{supp}(W_{f_{kr}}),\,\gamma_k\cdot\omega''\oplus a_{kr}\oplus a'_{kr}=0\},\\&\mathrm{supp}(W_{g'_{k}})=\bigcup_{r=0}^{2^p-1}\{(\omega',\omega'')\,|\,\omega'\in\operatorname{supp}(W_{f_{kr}}),\,\gamma_k\cdot\omega''\oplus b_{kr}\oplus b'_{kr}=0\}.
    \end{aligned}
\end{equation}

Fix \(k\). We first show that \(g_{k}\) and \(g'_{k}\) have
disjoint Walsh supports. Compare the part indexed by \(r\) in
the support representation of \(g_{k}\) with the part indexed by
\(r'\) in that of \(g_{k}'\). If \(r\neq r'\), these two parts
are disjoint because
\(\operatorname{supp}(W_{f_{kr}})
\cap\operatorname{supp}(W_{f_{kr'}})=\varnothing\).
Suppose now that \(r=r'\). If the two corresponding parts had
a common point, that point would satisfy both defining
conditions involving \(\omega''\). Taking the XOR of these two conditions
would yield
\(a_{kr}\oplus a'_{kr}\oplus b_{kr}\oplus b'_{kr}=0\),
contradicting the defining equality
\(a_{kr}\oplus a'_{kr}\oplus b_{kr}\oplus b'_{kr}=1\).
Thus every part in the first union is disjoint from every part
in the second union. Taking the unions, we obtain \begin{equation*}
    \operatorname{supp}(W_{g_{k}})
\cap\operatorname{supp}(W_{g'_{k}})=\varnothing.
\end{equation*}

It remains to compare functions corresponding to distinct indices.
Let \(0\leq k <k'\leq 2^{q}-1\).
Suppose that there exist
\(h\in\{g_{k},g'_{k}\}\) and
\( h'\in\{g_{k'},g'_{k'}\}\) such that
\(\operatorname{supp}(W_h)\cap
\operatorname{supp}(W_{h'})\neq\varnothing\).
By the support descriptions established above, there would then
exist \(r,r'\) satisfying $0\leq r,r'\leq 2^p-1$ and
\(\omega'\in\mathbb F_2^n\) such that
\(\omega'\in\operatorname{supp}(W_{f_{kr}})
\cap\operatorname{supp}(W_{f_{k'r'}})\).
Since $k\ne k'$, the pairs $(k,r)$ and $(k',r')$ are
distinct, contradicting the assumption that the constituent
functions form a disjoint spectra family. Hence,
\[
\operatorname{supp}(W_h)\cap
\operatorname{supp}(W_{h'})=\varnothing
\quad\text{for all}\quad
h\in\{g_{k},g'_{k}\},\ \
h'\in\{g_{k'},g'_{k'}\}.
\]

Combining the two cases, $\Omega_R^{\mathcal F}$ is a family
of disjoint spectra $(q+1)$-plateaued functions. Each
$(q+1)$-plateaued function in $n+p+1$ variables has Walsh
support of size $2^{n+p-q}$. Consequently, any disjoint
spectra family of such functions has at most
$2^{n+p+1}/2^{n+p-q}=2^{q+1}$ members. Since
$|\Omega_R^{\mathcal F}|=2^{q+1}$, this bound is attained.
\end{proof}

\section{Families of disjoint spectra plateaued functions
without nonzero linear structures based on linear functions}\label{sec-plateauedwithoutlinear}
The construction in Proposition~\ref{The-disjoint} yields a
maximum-cardinality family of disjoint spectra plateaued Boolean functions. We now impose additional conditions to ensure
that no member of the family admits a nonzero linear structure.

\subsection{A no-hyperplane partition and shifted nonconstancy sequences}\label{subse-partitionDk}
\begin{Definition}
A subset of $\mathbb F_2^m$ is said to have the
no-hyperplane property if it is not contained in any affine
hyperplane of $\mathbb F_2^m$.
\end{Definition}
We next construct a no-hyperplane partition and shifted
nonconstancy sequences, which will be used to exclude linear
structures in directions with
$\beta=\boldsymbol{0}_{p+1}$ and $\beta=\gamma_k$, respectively.

Let \(p\geq 2\), \(0\leq q<2^p-p-1\), and \(m=p+q\). Choose Boolean functions \(\phi_0,\ldots,\phi_{q-1}\in\mathcal{B}_p\) such that
\(1,y_0,\ldots,y_{p-1},\phi_0,\ldots,\phi_{q-1}\) is linearly independent in \(\mathcal{B}_p\), and put
\begin{equation}
U=\operatorname{span}
\{1,y_0,\ldots,y_{p-1},
  \phi_0,\ldots,\phi_{q-1}\}.
\label{spanU}
\end{equation}
Define \(h:\mathbb{F}_2^p\to\mathbb{F}_2^q\) by
    $h(y')=(\phi_0(y'),\ldots,\phi_{q-1}(y'))$.
 For every \(k\in\mathbb{F}_2^q\), let
     $\mathcal{D}_k=
\{(y',h(y')\oplus k)\,|\,y'\in\mathbb{F}_2^p\}
\subseteq\mathbb{F}_2^m$.
When \(q=0\), the list $\phi_0,\ldots,\phi_{q-1}$ is understood
to be empty, \(h\) is interpreted as the unique map into \(\mathbb{F}_2^0\), and \(\mathcal{D}_0=\mathbb{F}_2^p\).

The assumed bound on \(q\) guarantees that this choice is possible. Indeed, let
$\mathcal{AF}_p=\operatorname{span}\{1,y_0,\ldots,y_{p-1}\}$.
 Then
\(\dim(\mathcal{B}_p/\mathcal{AF}_p)=2^p-p-1\). Hence the residue classes of
\(\phi_0,\ldots,\phi_{q-1}\) may be chosen linearly independent modulo \(\mathcal{AF}_p\). Moreover, \(\dim U=p+q+1<2^p=\dim\mathcal{B}_p\), so \(U\) is a proper subspace of \(\mathcal{B}_p\).

Fix \(v_0,\ldots,v_{2^p-1}\) such that
 $\{v_r\}_{r=0}^{2^p-1} =\mathbb{F}_2^p$. For every \(k\in\mathbb{F}_2^q\) and \(0\leq r\leq 2^p-1\), define
\begin{equation}\label{ckr}
    c_{kr}=(v_r,h(v_r)\oplus k)\in \mathcal{D}_k.
\end{equation}
Thus \begin{equation}\label{Eq-Dk}
\mathcal{D}_k=\{c_{kr}\,|\,0\leq r\leq 2^p-1\}.
\end{equation} For
\(\alpha\in\mathbb{F}_2^m\), define
   $ \mathcal{C}_k(\alpha)=
\bigl(\alpha\cdot c_{k0},\ldots,
      \alpha\cdot c_{k,2^p-1}\bigr)$.

Since \(U\subsetneq\mathcal{B}_p\), choose
\(\psi_k\in\mathcal{B}_p\setminus U\) and denote its truth
table, with respect to the fixed vectors \(v_r\),
\(0\leq r\leq 2^p-1\), by 
\begin{equation*}
    \Psi_k=\bigl(\psi_k(v_0),\ldots,
\psi_k(v_{2^p-1})\bigr).
\end{equation*}
Because \(U\) is independent of \(k\), the same function
\(\psi\in\mathcal{B}_p\setminus U\) may be used for all \(k\).

For a concrete realization of the above construction, choose $\phi_0,\ldots,\phi_{q-1}$ as any $q$ distinct ANF monomials whose degrees lie between $2$ and $p-1$, and set
$\psi_k(y')=\psi(y')=\prod_{i=0}^{p-1}y_i
$ for every $k$.
Such a choice is possible under the standing assumption $q<2^p-p-1$. Since $\psi$ has degree $p$, whereas every function in $U$ has degree at most $p-1$, we have $\psi\notin U$. Thus, the same function $\psi$ can be used for all $k$.

\begin{Proposition}
\label{Pro-nonconstant}
For every \(k\in\mathbb{F}_2^q\), the following statements hold.
\begin{enumerate}
\item The sets \(\mathcal{D}_k\), \(k\in\mathbb{F}_2^q\), form a partition of
\(\mathbb{F}_2^m\), and every \(\mathcal{D}_k\) has the no-hyperplane property. Equivalently, \(C_k(\alpha)\) is nonconstant for every
\(\alpha\in\mathbb{F}_2^m\setminus\{\boldsymbol{0}_m\}\).

\item For every \(\alpha\in\mathbb{F}_2^m\), the sequence
\(C_k(\alpha)\oplus\Psi_k\) is nonconstant.
\end{enumerate}
\end{Proposition}

\begin{proof}
\noindent\textbf{(1).} Every \((y',u)\in\mathbb{F}_2^p\times\mathbb{F}_2^q\) belongs to \(\mathcal{D}_k\) for the unique vector \(k=u\oplus h(y')\). Thus the sets \(\mathcal{D}_k\) form a partition of \(\mathbb{F}_2^m\).

Suppose, to the contrary, that $\mathcal{D}_k$ is contained in a
proper affine hyperplane of $\mathbb{F}_2^m$. Then there exist
$\alpha=(\alpha_0,\alpha_1)\in
\mathbb{F}_2^p\times\mathbb{F}_2^q\setminus\{\mathbf{0}_m\}$
and $c\in\mathbb{F}_2$ such that
$\alpha_0\cdot y'
\oplus
\alpha_1\cdot\bigl(h(y')\oplus k\bigr)
\oplus c
=0$
for every $y'\in\mathbb{F}_2^p$.
Equivalently,
$\alpha_0\cdot y'
\oplus
\bigoplus_{i=0}^{q-1}\alpha_{1i}\phi_i(y')
\oplus
(\alpha_1\cdot k\oplus c)
=0$
for all $y'\in\mathbb{F}_2^p$.
By the linear independence of
$1,y_0,\ldots,y_{p-1},\phi_0,\ldots,\phi_{q-1}$,
all coefficients in the above identity must vanish. In particular,
$\alpha_0=\mathbf{0}_p$ and
$\alpha_1=\mathbf{0}_q$,
so that $\alpha=\mathbf{0}_m$, contradicting
$\alpha\neq\mathbf{0}_m$.
Hence $\mathcal{D}_k$ is not contained in any proper affine hyperplane.
Equivalently, $C_k(\alpha)$ is nonconstant for every
$\alpha\in\mathbb{F}_2^m\setminus\{\mathbf{0}_m\}$.

\noindent\textbf{(2).} \(C_k(\alpha)\) is the truth table of
$f_{\alpha,k}(y')
=
\alpha_0\cdot y'
\oplus
\alpha_1\cdot h(y')
\oplus
\alpha_1\cdot k$
and \(f_{\alpha,k}\in U\). If \(C_k(\alpha)\oplus\Psi_k\) were constant, then
\(f_{\alpha,k}\oplus\psi_k\) would be a constant Boolean function. Since \(f_{\alpha,k}\in U\) and \(1\in U\), this would imply
\(\psi_k\in U\), contrary to $\psi_k\notin U$. Therefore
\(C_k(\alpha)\oplus\Psi_k\) is nonconstant for every
\(\alpha\in\mathbb{F}_2^m\).
\end{proof}

\subsection{Explicit construction of families of disjoint spectra plateaued
functions without nonzero linear structures based on linear functions}\label{4.2}
 We now specialize the general construction by taking the
constituent functions to be linear and using the no-hyperplane
partition and shifted nonconstancy sequences developed above.
These choices will allow us to exclude nonzero linear structures
from every member of the resulting family. First, we align the indexing of the selector subspaces with the vectors
$\{v_r\}_{r=0}^{2^p-1}$ introduced in Subsection \ref{subse-partitionDk}.
For each $k$, choose a linear isomorphism
\begin{equation}\label{Iso-A}
    \mathcal{L}_k:\mathbb{F}_2^p\rightarrow S_k,
\end{equation}
and, since the ordering of the elements of $S_k$ is immaterial in the
preceding abstract construction, we henceforth take
$s_{kr}=\mathcal{L}_k(v_r)$ for
$0\leq r\leq 2^p-1$.
Consequently,
$s'_{kr}
=
s_{kr}\oplus\gamma_k
=
\mathcal{L}_k(v_r)\oplus\gamma_k$.

 Suppose $p\geq 2$, $0\leq q<2^p-p-1$ and set
$n=m=p+q$.
  Using vectors defined in \eqref{ckr}, specialize the family in \eqref{Set-F} to
\begin{equation}\label{Set-Fckr}
\begin{aligned}
   \mathcal{F}_0 =& \bigl\{f_{kr}\,|\,f_{kr}(x)=c_{kr}\cdot x, 0\leq k\leq 2^q-1, 0\leq r\leq 2^p-1  \bigr\}
\end{aligned}
\end{equation} where  $x\in\F_2^m$. Since $\{\mathcal{D}_k\}_{k\in\F_2^q}$ is a partition of $\F_2^m$, the
collection
$\{c_{kr}\,|\,0\leq k\leq 2^q-1,\ 0\leq r\leq 2^p-1\}$
exhausts $\F_2^m$. Consequently, $\mathcal{F}_0$ is a family of disjoint spectra
$m$-plateaued Boolean functions in $m$ variables.

The resulting functions $g_k$ and $g_k'$ belong to $\mathcal{GMM}$. Indeed, by
\eqref{Eq-g_k} and the choice
$f_{kr}(x)=c_{kr}\cdot x$, define
    $\pi_k(y)=c_{k,\sigma_k(y)}$.
Then
\begin{equation*}
\begin{aligned}
       &g_{k}(x,y)
=
x\cdot\pi_k(y)\oplus a_{k}(y),\\
&g'_{k}(x,y)
=
x\cdot\pi_k(y)\oplus b_{k}(y).
\end{aligned}
\end{equation*}
Moreover, since the vectors $c_{kr}$ are distinct and
$\sigma_k^{-1}(r)=\{s_{kr},s'_{kr}\}$,
we have
$\pi_k^{-1}(c_{kr})
=
\{s_{kr},s'_{kr}\}$
for every $0\le r\le 2^p-1$. Hence,
$\pi_k$ is two-to-one onto $\mathcal{D}_k$. Therefore, both
$g_{k}$ and $g'_{k}$ are $\mathrm{GMM}$ functions.
  
For each $0 \leq k \leq 2^q-1$, choose
$\psi_k \in \mathcal{B}_p \setminus U$, where $U$ is defined in
\eqref{spanU}. Then, for every $0 \leq r \leq 2^p-1$, choose
$a_{kr},a'_{kr}\in\F_2$ such that
$a_{kr}\oplus a'_{kr}=\psi_k(v_r)$.
Using the functions $\psi_k$, we now specialize the offset family
$R$ introduced in \eqref{Set-R} as follows:
\begin{equation}\label{Set-R(psi)}
R
=
\left(
a_{kr},a'_{kr},b_{kr},b'_{kr}
\right)_{\substack{
0\leq k\leq 2^q-1\\
0\leq r\leq 2^p-1
}},
\end{equation} where $a_{kr}\oplus a'_{kr}=\psi_k(v_r)$,
$b_{kr}\oplus b'_{kr}
=\psi_k(v_r)\oplus 1$ for $0\leq k\leq 2^q-1$, $0\leq r\leq 2^p-1$.
These two relations imply the  condition in
\eqref{Set-R}.

All ingredients of the construction are now in place. For ease
of reference, the procedure leading to the family
$\Omega_{\mathcal R}^{(\mathcal{F}_0)}$ is summarized in
Table~\ref{tab:construction-roadmap}.
\begin{table}[!ht]
\caption{Roadmap for the explicit construction  $\Omega_{R}^{\mathcal{F}_0}$  in Theorem~\ref{The-Set-withoutlinear}.}
\label{tab:construction-roadmap}
\centering
\setlength{\tabcolsep}{5pt}
\renewcommand{\arraystretch}{1.08}

\begin{tabularx}{\linewidth}
{@{}>{\bfseries}l >{\raggedright\arraybackslash}X@{}}
\toprule

Step 1. &
Choose integers $p\geq2$ and $0\leq q<2^p-p-1$, and set
$m=p+q$. Fix \(v_0,\ldots,v_{2^p-1}\) such that
\(\{v_r\}_{r=0}^{2^p-1}=\mathbb F_2^p\).
\\
\addlinespace[3pt]

Step 2. &
Choose $\phi_0,\ldots,\phi_{q-1}\in\mathcal B_p$ such
that $1,y_0,\ldots,y_{p-1},\phi_0,\ldots,\phi_{q-1}$
are linearly independent, and define $U$ and $h$.
\\
\addlinespace[3pt]

Step 3. &
Define $c_{kr}$, $\mathcal{D}_k$, $f_{kr}$, and the constituent family
$\mathcal {F}_0$ of disjoint spectra linear functions.
\\
\addlinespace[3pt]

Step 4. &
Choose a $p$-dimensional subspace
$S_k\subsetneq\F_2^{p+1}$, a vector $\gamma_k\notin S_k$, and a
linear isomorphism
$\mathcal L_k:\F_2^p\rightarrow S_k$.
Set $s_{kr}=\mathcal L_k(v_r)$ and construct $\sigma_k$ from
the pairs $\{s_{kr},s_{kr}\oplus\gamma_k\}$.
\\
\addlinespace[3pt]

Step 5. &
Choose $\psi_k\in\mathcal B_p\setminus U$ and $R$ with
$a_{kr}\oplus a'_{kr}=\psi_k(v_r)$ and
$b_{kr}\oplus b'_{kr}=\psi_k(v_r)\oplus1$.
\\
\addlinespace[3pt]

Step 6. &
Define the offset functions $a_k,b_k$, then define $g_k,g'_k$
and form $\Omega_R^{\mathcal {F}_0}$.
\\

\bottomrule
\end{tabularx}

\vspace{2pt}
\parbox{\linewidth}{%
\footnotesize\raggedright
\textit{Note:}
Throughout the construction,
$0\leq k\leq2^q-1$ and $0\leq r\leq2^p-1$.
}
\end{table}
\FloatBarrier
The properties of the resulting family are established in the
following theorem.
\FloatBarrier
\begin{Theorem}\label{The-Set-withoutlinear}
       The family  $\Omega_{R}^{\mathcal{F}_0}$ defined in \eqref{Set-Omega} is a maximum-cardinality family of disjoint spectra
\((q+1)\)-plateaued $\mathrm{GMM}$ functions in \(m+p+1\) variables without nonzero linear structures, where $\mathcal{F}_0$ is defined in $\eqref{Set-Fckr}$ and $R$ is defined in \eqref{Set-R(psi)}.
\end{Theorem}
 \begin{proof}
By Proposition~\ref{The-disjoint}, the family is a maximum-cardinality family of
$2^{q+1}$ disjoint spectra $(q+1)$-plateaued functions in $m+p+1$ variables. The representation established above shows that every member
of $\Omega_R^{\mathcal{F}_0}$ is a $\mathrm{GMM}$ function. It remains to prove the absence of
nonzero linear structures.

Fix $k$ with $0\leq k\leq  2^q-1$.
We now consider $D_{(\alpha,\beta)}g_k$ for all nonzero pairs
$(\alpha,\beta)\in
\mathbb F_2^m\times\mathbb F_2^{p+1}$.

\noindent\textbf{Case 1:}
$\beta\notin\{\boldsymbol{0}_{p+1},\gamma_k\}$,
$\alpha\in\mathbb{F}_2^m$.

From the definition of $g_{k}$, its derivative equals
\begin{equation*}
    \begin{aligned}
D_{(\alpha,\beta)}g_{k}(x,y)
=&f_{k,\sigma_k(y)}(x)\oplus f_{k,\sigma_k(y\oplus\beta)}(x\oplus\alpha)\oplus a_{k}(y)\oplus a_{k}(y\oplus\beta)
\\=&\bigl(
c_{k,\sigma_k(y)}
\oplus
c_{k,\sigma_k(y\oplus\beta)}
\bigr)\cdot x
\oplus
c_{k,\sigma_k(y\oplus\beta)}\cdot\alpha
\oplus
a_{k}(y)
\oplus
a_{k}(y\oplus\beta),
    \end{aligned}
\end{equation*}

where $(x,y)\in\mathbb{F}_2^m\times\mathbb{F}_2^{p+1}$.
We first show that
$\sigma_k(y)\neq\sigma_k(y\oplus\beta)$ for every
$y\in\mathbb{F}_2^{p+1}$. Suppose, to the contrary, that
$\sigma_k(y)=\sigma_k(y\oplus\beta)=r$ for some $r$.
By the definition of $\sigma_k$, both $y$ and $y\oplus\beta$
belong to
$\{s_{kr},s'_{kr}\}
=
\{s_{kr},s_{kr}\oplus\gamma_k\}$.
Hence their difference $\beta$ must be either
$\boldsymbol{0}_{p+1}$ or $\gamma_k$, contrary to the assumption
$\beta\notin\{\boldsymbol{0}_{p+1},\gamma_k\}$. Therefore,
$\sigma_k(y)\neq\sigma_k(y\oplus\beta)$ for every
$y\in\mathbb{F}_2^{p+1}$.

Since the vectors $c_{kr}$ are distinct for distinct values of $r$,
it follows that
$c_{k,\sigma_k(y)}
\neq
c_{k,\sigma_k(y\oplus\beta)}$
for every $y$. Thus,
\[
c_{k,\sigma_k(y)}
\oplus
c_{k,\sigma_k(y\oplus\beta)}
\neq
\boldsymbol{0}_m.
\]
Therefore, $D_{(\alpha,\beta)}g_{k}(x,y)$ has a nonzero
coefficient of $x$ and is nonconstant as a function of $x$.
Consequently, $(\alpha,\beta)$ is not a linear structure of
$g_{k}$.

The same argument applies to $g'_{k}$, since $g_{k}$ and
$g'_{k}$ have the same constituent functions and differ only
in their offset parameters.

\noindent\textbf{Case 2:}
$\beta=\boldsymbol{0}_{p+1}$,
$\alpha\neq\boldsymbol{0}_m$.

The derivatives of $g_{k}$ and $g'_{k}$ are equal and given by
\begin{equation*}
    \begin{aligned}
    D_{(\alpha,\boldsymbol{0}_{p+1})}g_{k}(x,y)
=
D_{(\alpha,\boldsymbol{0}_{p+1})}g'_{k}(x,y)
=&f_{k,\sigma_k(y)}(x)\oplus f_{k,\sigma_k(y)}(x\oplus\alpha)\\=&
c_{k,\sigma_k(y)}\cdot x\oplus c_{k,\sigma_k(y)}\cdot (x\oplus\alpha)
\\=&
\alpha\cdot c_{k,\sigma_k(y)}.
    \end{aligned}
\end{equation*}
Since $\sigma_k$ is surjective, as $y$ ranges over
$\mathbb{F}_2^{p+1}$, the values of this derivative are precisely
\[
\alpha\cdot c_{k0},\;
\alpha\cdot c_{k1},\;
\ldots,\;
\alpha\cdot c_{k,2^p-1},
\]
namely, the entries of $C_k(\alpha)$. By Proposition~\ref{Pro-nonconstant}(1), the
no-hyperplane property implies that $C_k(\alpha)$ is nonconstant
for every $\alpha\neq\boldsymbol{0}_m$. Therefore,
$D_{(\alpha,\boldsymbol{0}_{p+1})}g_{k}$ and
$D_{(\alpha,\boldsymbol{0}_{p+1})}g'_{k}$ are nonconstant, so
$(\alpha,\boldsymbol{0}_{p+1})$ is not a linear structure of
either $g_{k}$ or $g'_{k}$.

\noindent\textbf{Case 3:}
$\beta=\gamma_k$, $\alpha\in\mathbb{F}_2^m$.

For every $y\in\mathbb{F}_2^{p+1}$, translation by $\gamma_k$
interchanges the two elements $s_{kr}$ and $s'_{kr}$ of the pair
indexed by $r=\sigma_k(y)$. Hence,
$\sigma_k(y\oplus\gamma_k)=\sigma_k(y)$. Moreover,
$a_{k}(y)\oplus a_{k}(y\oplus\gamma_k)
=
a_{k,\sigma_k(y)}\oplus a'_{k,\sigma_k(y)}$.
Then, the derivative of $g_{k}$ equals
\[
\begin{aligned}
D_{(\alpha,\gamma_k)}g_{k}(x,y)
&=
f_{k,\sigma_k(y)}(x)
\oplus
f_{k,\sigma_k(y)}(x\oplus\alpha)
\oplus
a_{k}(y)
\oplus
a_{k}(y\oplus\gamma_k)\\
&=
\alpha\cdot c_{k,\sigma_k(y)}
\oplus
a_{k,\sigma_k(y)}
\oplus
a'_{k,\sigma_k(y)}.
\end{aligned}
\]
Similarly, we have
\begin{equation*}
D_{(\alpha,\gamma_k)}g'_{k}(x,y)
=
\alpha\cdot c_{k,\sigma_k(y)}
\oplus
b_{k,\sigma_k(y)}
\oplus
b'_{k,\sigma_k(y)}.
\end{equation*}

Using
$a_{kr}\oplus a'_{kr}=\psi_k(v_r)$ and
$b_{kr}\oplus b'_{kr}=\psi_k(v_r)\oplus1$,
the value sequences of these two derivatives, indexed by
$0\leq r\leq 2^p-1$, are respectively
$C_k(\alpha)\oplus\Psi_k$
and
$C_k(\alpha)\oplus\Psi_k
\oplus
\boldsymbol{1}_{2^p}$,
where $\boldsymbol{1}_{2^p}$ denotes the all-one vector of length $2^p$.
By Proposition~\ref{Pro-nonconstant}(2), the first sequence is nonconstant for every
$\alpha\in\mathbb{F}_2^m$. Since the second is its bitwise
complement, it is also nonconstant. Moreover, $\sigma_k$ is
surjective, so neither derivative is constant. Consequently,
$(\alpha,\gamma_k)$ is not a linear structure of either
$g_{k}$ or $g'_{k}$ for any
$\alpha\in\mathbb{F}_2^m$.

Combining the three cases, neither $g_{k}$ nor $g'_{k}$ admits
a nonzero linear structure. Since $k$ was arbitrary, no member of
$\Omega_{R}^{\mathcal {F}_0}$ admits a nonzero linear structure.
\end{proof}

\subsection{Structural distinction from translation-based families}
Theorem~\ref{The-Set-withoutlinear} yields a maximum-cardinality
family of disjoint spectra plateaued functions without  nonzero linear structures.
We now compare the Walsh-support structure of the family constructed above
 with the spectral construction of
\cite[Theorem~4.4]{disjoint2}. In that construction, if
\begin{equation}\label{fff}
    \{f_i\mid 0\leq i\leq 2^s-1\}
\end{equation}
denotes the resulting maximum-cardinality family of disjoint spectra plateaued functions without nonzero linear structures, then
\[
\operatorname{supp}(W_{f_i})
=
\operatorname{supp}(W_f)\oplus q_i,
\qquad
0\leq i\leq 2^s-1,
\]
where $\operatorname{supp}(W_f)$ is fixed and $q_i\in Q$, with
$Q=\{\boldsymbol{0}_{n-s}\}\times\mathbb F_2^s
   \subseteq\mathbb F_2^n$.
Thus, all Walsh supports in $\{f_i\mid 0\leq i\leq 2^s-1\}$ are translations of the
same set $\operatorname{supp}(W_f)$.

Let a common EA transformation of the family in \eqref{fff} be given by
\[
\widetilde{f}_i(x)
=
f_i(Ax\oplus a)\oplus b\cdot x\oplus c,
\qquad 0\leq i\leq 2^s-1,
\]
where $A\in GL(n,2)$, $a,b\in\mathbb{F}_2^n$, and
$c\in\mathbb{F}_2$. By
Proposition~\ref{Prop-linear-equivalence-invariants},
\begin{equation*}
\operatorname{supp}(W_{\widetilde{f}_i})
=
A^T\operatorname{supp}(W_{f_i})\oplus b
=
A^T\bigl(\operatorname{supp}(W_f)\oplus q_i\bigr)\oplus b
=
\bigl(A^T\operatorname{supp}(W_f)\oplus b\bigr)
\oplus A^Tq_i.
\end{equation*}
Thus, the transformed family
$\{\widetilde{f}_i\mid 0\leq i\leq 2^s-1\}$
is again a family of plateaued functions without nonzero linear
structures whose pairwise disjoint Walsh supports satisfy the
common-translation condition.

In contrast, Theorem~\ref{The-Set-withoutlinear} can produce
maximum-cardinality families that are not commonly EA-equivalent
to any family obtained from the construction above. The following
proposition gives a sufficient criterion for distinguishing the two
types of families under common EA-equivalence.
\begin{Proposition}\label{Prop-direct-separation}
Let
$\Omega_{R}^{\mathcal{F}_0}$
be the family determined by $\mathcal{F}_0$ in~\eqref{Set-Fckr}
and $R$ in~\eqref{Set-R(psi)}.
Let $0\leq k<k'\leq 2^q-1$ and suppose that
$\gamma_k\neq\gamma_{k'}$.
Then, for any
$h\in\{g_{k},g'_{k}\}$ and
$h'\in\{g_{k'},g'_{k'}\}$,
there exists no
$\tau\in\mathbb{F}_2^{m+p+1}$ such that
\begin{equation}\label{Wh}
\operatorname{supp}(W_h)
=
\tau\oplus\operatorname{supp}(W_{h'}).
\end{equation}
\end{Proposition}
\begin{proof}
For each $0\leq r\leq 2^p-1$, since
$f_{kr}(x)=c_{kr}\cdot x$, we have
\[
W_{f_{kr}}(\omega')
=
\sum_{x\in\mathbb{F}_2^m}
(-1)^{f_{kr}(x)\oplus \omega'\cdot x}
=
\sum_{x\in\mathbb{F}_2^m}
(-1)^{(c_{kr}\oplus\omega')\cdot x}
=
\begin{cases}
2^m, & \omega'=c_{kr},\\
0, & \omega'\neq c_{kr}.
\end{cases}
\]
Hence, 
$\operatorname{supp}(W_{f_{kr}})=\{c_{kr}\}$.
By~\eqref{Suppgg'} and~\eqref{Set-R(psi)},
\begin{equation}\label{suppgg}
    \begin{aligned}
\operatorname{supp}(W_{g_{k}})
&=
\bigcup_{r=0}^{2^p-1}
\left\{
(c_{kr},\omega'')
\,\middle|\,
\gamma_k\cdot\omega''=\psi_k(v_r)
\right\},\\
\operatorname{supp}(W_{g'_{k}})
&=
\bigcup_{r=0}^{2^p-1}
\left\{
(c_{kr},\omega'')
\,\middle|\,
\gamma_k\cdot\omega''=\psi_k(v_r)\oplus1
\right\}.
\end{aligned}
\end{equation}
The corresponding expressions for $g_{k'}$ and $g'_{k'}$ are
obtained by replacing $k$ with $k'$.

Let $h\in\{g_{k},g'_{k}\}$ and
$h'\in\{g_{k'},g'_{k'}\}$.
Suppose, to the contrary, that there exists
$\tau=(\tau',\tau'')
\in\mathbb{F}_2^m\times\mathbb{F}_2^{p+1}$
such that \eqref{Wh} holds.

Fix $0\leq r\leq 2^p-1$. Since the fiber of
$\operatorname{supp}(W_h)$ with first component $c_{kr}$ is nonempty,
\eqref{Wh} implies that
$c_{kr}\oplus\tau'\in\mathcal{D}_{k'}$.
Hence there exists a unique
$0\leq r'\leq 2^p-1$ such that
$c_{kr}\oplus\tau'=c_{k'r'}$.
Consequently, for every
$\omega''\in\mathbb{F}_2^{p+1}$,
$(c_{kr},\omega'')\in\operatorname{supp}(W_h)$ if and only if
$(c_{k'r'},\omega''\oplus\tau'')
\in\operatorname{supp}(W_{h'})$.

We now consider the four possible choices of $h$ and $h'$.

\noindent
\textbf{Case 1:} $h=g_{k}$ and $h'=g_{k'}$.

The above equivalence gives
\[
\{
\omega''\in\mathbb{F}_2^{p+1}
\,|\,
\gamma_k\cdot\omega''=\psi_k(v_r)
\}
=
\{
\omega''\in\mathbb{F}_2^{p+1}
\,|\,
\gamma_{k'}\cdot\omega''
=
\psi_{k'}(v_{r'})
\oplus\gamma_{k'}\cdot\tau''
\}.
\]
For this fixed $r$ and the corresponding $r'$, the two sets are
nonempty affine hyperplanes in $\mathbb{F}_2^{p+1}$ with normal
vectors $\gamma_k$ and $\gamma_{k'}$, respectively. Their equality implies that the corresponding linear hyperplanes
coincide, and hence
$\operatorname{span}(\gamma_k)
=
\operatorname{span}(\gamma_{k'})$.
Since $\gamma_k,\gamma_{k'}\neq0$, we obtain
$\gamma_k=\gamma_{k'}$, a contradiction.

\noindent
\textbf{Case 2:} $h=g_{k}$ and $h'=g'_{k'}$.

In this case,
\[
\{
\omega''\in\mathbb{F}_2^{p+1}
\,|\,
\gamma_k\cdot\omega''=\psi_k(v_r)
\}
=
\{
\omega''\in\mathbb{F}_2^{p+1}
\,|\,
\gamma_{k'}\cdot\omega''
=
\psi_{k'}(v_{r'})\oplus1
\oplus\gamma_{k'}\cdot\tau''
\}.
\]
For this fixed $r$ and the corresponding $r'$, equality of the two
affine hyperplanes again implies
$\gamma_k=\gamma_{k'}$, a contradiction.

\noindent
\textbf{Case 3:} $h=g'_{k}$ and $h'=g_{k'}$.

We obtain
\[
\{
\omega''\in\mathbb{F}_2^{p+1}
\,|\,
\gamma_k\cdot\omega''=\psi_k(v_r)\oplus1
\}
=
\{
\omega''\in\mathbb{F}_2^{p+1}
\,|\,
\gamma_{k'}\cdot\omega''
=
\psi_{k'}(v_{r'})
\oplus\gamma_{k'}\cdot\tau''
\}.
\]
Hence, for this fixed $r$ and the corresponding $r'$, we obtain
$\gamma_k=\gamma_{k'}$, again a contradiction.

\noindent
\textbf{Case 4:} $h=g'_{k}$ and $h'=g'_{k'}$.

Finally,
\[
\{
\omega''\in\mathbb{F}_2^{p+1}
\,|\,
\gamma_k\cdot\omega''=\psi_k(v_r)\oplus1
\}
=
\{
\omega''\in\mathbb{F}_2^{p+1}
\,|\,
\gamma_{k'}\cdot\omega''
=
\psi_{k'}(v_{r'})\oplus1
\oplus\gamma_{k'}\cdot\tau''
\}.
\]
For this fixed $r$ and the corresponding $r'$, this again implies
$\gamma_k=\gamma_{k'}$, a contradiction.

Thus none of the four cases is possible. Therefore, no
$\tau\in\mathbb{F}_2^{m+p+1}$ can satisfy \eqref{Wh}.
\end{proof}
\begin{Remark}
By Proposition~\ref{Prop-direct-separation}, if there exist distinct
$k,k'\in\mathbb F_2^q$ such that
$\gamma_k\neq\gamma_{k'}$, then the Walsh supports in
$\Omega_R^{\mathcal{F}_0}$ are not all translations of a single
fixed support. Hence, this family does not satisfy the
common-translation condition. Every family obtained from
\cite[Theorem~4.4]{disjoint2} satisfies this condition, which is
preserved under a common EA transformation. Consequently,
$\Omega_R^{\mathcal{F}_0}$ is not commonly EA-equivalent to any
family obtained from that construction. When $q=0$, the two Walsh
supports are translations of each other, and hence this criterion does not distinguish the resulting family.
\end{Remark}

The preceding result rules out common EA-equivalence. We next
consider the broader notion of memberwise EA-equivalence.

 Suppose that, for each $0\leq i\leq 2^s-1$, an EA
transformation with parameters $A_i\in\operatorname{GL}(n,2)$,
$a_i,b_i\in\mathbb F_2^n$, and $c_i\in\mathbb F_2$ is chosen
independently, and set
$\widehat f_i(x)=f_i(A_i x\oplus a_i)\oplus b_i\cdot x\oplus c_i$,
where $\{f_i\mid 0\leq i\leq 2^s-1\}$ is the family given by~\eqref{fff}. Assume that the resulting family
$\{\widehat f_i\mid 0\leq i\leq 2^s-1\}$
is again a disjoint spectra family.
 The following result gives
a necessary condition satisfied by every such family.
\begin{Proposition}\label{Pro-5}
For every $0\leq i\leq 2^s-1$,
\[
|
\operatorname{supp}(W_{\widehat f_i})
\oplus
\operatorname{supp}(W_{\widehat f_i})
|
=
|
\operatorname{supp}(W_f)
\oplus
\operatorname{supp}(W_f)
|.
\]
In particular, the Walsh-support self-sum cardinality is constant over $\{\widehat f_i\,|\, 0\leq i\leq 2^s-1\}$.
\end{Proposition}

\begin{proof}
By Proposition~1 and
$\operatorname{supp}(W_{f_i})
=\operatorname{supp}(W_f)\oplus q_i$, we have
$\operatorname{supp}(W_{\widehat f_i})
=A_i^T(\operatorname{supp}(W_f)\oplus q_i)\oplus b_i$.
Consequently,
\begin{equation*}
   \begin{aligned}
\operatorname{supp}(W_{\widehat f_i})
 \oplus
 \operatorname{supp}(W_{\widehat f_i})
&=
\bigl(
 A_i^T(\operatorname{supp}(W_f)\oplus q_i)\oplus b_i
\bigr)
\oplus
\bigl(
 A_i^T(\operatorname{supp}(W_f)\oplus q_i)\oplus b_i
\bigr)\\
&=
A_i^T\bigl(
 \operatorname{supp}(W_f)
 \oplus
 \operatorname{supp}(W_f)
\bigr).
\end{aligned}
\end{equation*}
 Because $A_i^T$ is
bijective, it preserves cardinality, which proves the assertion.
\end{proof}
Proposition~\ref{Pro-5} provides a directly computable criterion for ruling out
memberwise EA-equivalence to any family obtained from \cite[Theorem~4.4]{disjoint2}. In particular, it suffices to find two
members of a family whose Walsh-support self-sums have different
cardinalities. Example~\ref{Ex-optimal-family} below shows that this criterion can be
satisfied by the construction in Theorem \ref{The-Set-withoutlinear}.

\subsection{Coordinate forms of the constructed $\mathrm{GMM}$ functions}\label{subse-GMMrep}

 The preceding construction already yields $\mathrm{GMM}$ functions. We now
introduce coordinates adapted to the parametrizations of
$\mathcal{D}_k$ and $S_k$. This will yield
 coordinate forms for the constructed functions and
make their algebraic degrees transparent.

Let
$\rho:\F_2^p\rightarrow\{0,1,\ldots,2^p-1\}$
be the index map associated with the fixed vectors
$\{v_r\}_{r=0}^{2^p-1}$; namely,
$\rho(v_r)=r$ for $0\leq r\leq2^p-1$.
For each $0\leq k\leq2^q-1$, associate with the offset family
$R$ in \eqref{Set-R(psi)} two Boolean functions $A_k,B_k\in\mathcal B_p$, defined by
$A_k(y')=a_{k,\rho(y')}$ and
$B_k(y')=b_{k,\rho(y')}$ for every $y'\in\F_2^p$.
Equivalently, $A_k(v_r)=a_{kr}$ and $B_k(v_r)=b_{kr}$.
Since $\{v_r\}_{r=0}^{2^p-1}=\F_2^p$, the functions
$A_k$ and $B_k$ are uniquely determined.

\begin{Proposition}
\label{Prop-GMM}
Let $g_{k},g'_{k}\in\Omega_{R}^{\mathcal {F}_0}$,
$0\le k\le 2^q-1$, where $\mathcal {F}_0$ is defined in \eqref{Set-Fckr}
and $R$ is defined in \eqref{Set-R(psi)}. Then there exists an
invertible linear transformation
$\mathcal T_k:\mathbb F_2^p\times\mathbb F_2
\to\mathbb F_2^{p+1}$ such that
\begin{equation}
\label{Eq-GMMREP}
\begin{aligned}
&\widetilde g_{k}(x,y',y_p)
=g_{k}\bigl(x,\mathcal T_k(y',y_p)\bigr)
=(y',h(y')\oplus k)\cdot x
  \oplus A_{k}(y')\oplus\psi_k(y')y_p,
\\&\widetilde g'_{k}(x,y',y_p)
=g'_{k}\bigl(x,\mathcal T_k(y',y_p)\bigr)
=(y',h(y')\oplus k)\cdot x
  \oplus B_{k}(y')\oplus\psi_k(y')y_p
  \oplus y_p,
\end{aligned}
\end{equation}
for every $(x,y',y_p)\in
\mathbb F_2^m\times\mathbb F_2^p\times\mathbb F_2$.
\end{Proposition}
\begin{proof}
Fix an integer $k$ with $0\leq k\leq2^q-1$.
Since $S_k$ is a $p$-dimensional subspace of
$\F_2^{p+1}$ and $\gamma_k\notin S_k$, we have
$\F_2^{p+1}=S_k\oplus\langle\gamma_k\rangle$.

Recall the linear isomorphism
$\mathcal L_k:\F_2^p\rightarrow S_k$ given by
\eqref{Iso-A}, where
$s_{kr}=\mathcal L_k(v_r)$ for
$0\leq r\leq2^p-1$. Define
\begin{equation}\label{Eq-isoT}
\mathcal T_k:\F_2^p\times\F_2\rightarrow\F_2^{p+1},
\qquad
\mathcal T_k(y',y_p)
=
\mathcal L_k(y')\oplus\gamma_k y_p.
\end{equation}
The direct-sum decomposition above implies that
$\mathcal T_k$ is a linear isomorphism.

For every $y'\in\F_2^p$, the definition of $\rho$ gives
$v_{\rho(y')}=y'$. Therefore,
$\mathcal L_k(y')=s_{k,\rho(y')}$, and hence
$\mathcal T_k(y',0)=s_{k,\rho(y')}$ and
$\mathcal T_k(y',1)=s'_{k,\rho(y')}$.
By the definition of the selector map $\sigma_k$ in \eqref{smap}, we obtain
\[
\sigma_k\bigl(\mathcal T_k(y',y_p)\bigr)=\rho(y')
\qquad
\text{for every }(y',y_p)\in\F_2^p\times\F_2.
\]

By the definitions of $A_k$ and $B_k$, we have
$A_k(y')=a_{k,\rho(y')}$ and
$B_k(y')=b_{k,\rho(y')}$. Moreover, the choice of the offset
family $R$ in \eqref{Set-R(psi)} gives
$a_{k,\rho(y')}\oplus a'_{k,\rho(y')}=\psi_k(y')$ and
$b_{k,\rho(y')}\oplus b'_{k,\rho(y')}=\psi_k(y')\oplus1$.
By \eqref{Set-Fckr},
$f_{k,\rho(y')}(x)
=c_{k,\rho(y')}\cdot x=(y',h(y')\oplus k)\cdot x$.

Using
$\sigma_k(\mathcal T_k(y',y_p))=\rho(y')$
in the definitions of $g_k$ and $g'_k$, we obtain
\[
\begin{aligned}
\widetilde g_k(x,y',y_p)
=g_k\bigl(x,\mathcal T_k(y',y_p)\bigr)
&=f_{k,\rho(y')}(x)
  \oplus a_{k,\rho(y')}
  \oplus
  \bigl(a_{k,\rho(y')}
        \oplus a'_{k,\rho(y')}\bigr)y_p\\
&=(y',h(y')\oplus k)\cdot x
  \oplus A_k(y')
  \oplus\psi_k(y')y_p,
\\[1mm]
\widetilde g'_k(x,y',y_p)
=g'_k\bigl(x,\mathcal T_k(y',y_p)\bigr)
&=f_{k,\rho(y')}(x)
  \oplus b_{k,\rho(y')}
  \oplus
  \bigl(b_{k,\rho(y')}
        \oplus b'_{k,\rho(y')}\bigr)y_p\\
&=(y',h(y')\oplus k)\cdot x
  \oplus B_k(y')
  \oplus\psi_k(y')y_p
  \oplus y_p.
\end{aligned}
\]
This proves \eqref{Eq-GMMREP}.
\end{proof}
\begin{Remark}
\label{Rem-individual-GMM-equivalence}
For each $k$, the same transformation $\mathcal T_k$
produces the coordinate forms in \eqref{Eq-GMMREP} for
both $g_{k}$ and $g'_{k}$. Its choice depends on the indexing
isomorphism $\mathcal L_k$ and the vector $\gamma_k$,
and may therefore vary with $k$.
\end{Remark}

The coordinate forms in \eqref{Eq-GMMREP} separate
the terms that determine the algebraic degrees of the
constructed functions. Since each $\mathcal{T}_k$ is invertible and
linear, algebraic degree is preserved under these coordinate
transformations. We now use these forms to prescribe a common algebraic degree for all family members.

\subsection{Algebraic-degree control for the constructed family}
\label{Subsec-optimal-degree-degenerate}
Let \(p\geq2\), \(3\leq d\leq p+1\), and \(m=p+q\), where
$0\leq q<
\sum_{i=2}^{d-1}\binom{p}{i}
\leq 2^p-p-1$.
In Subsection~\ref{subse-partitionDk}, choose Boolean
functions
\(\phi_0,\ldots,\phi_{q-1}\in\mathcal{B}_p\) such that
$1,y_0,\ldots,y_{p-1},
\phi_0,\ldots,\phi_{q-1}$
is linearly independent and
\(\deg(\phi_i)\leq d-1\) for every
\(0\leq i\leq q-1\). Let \(U\) be the subspace defined
in~\eqref{spanU}. For every \(0\leq k\leq2^q-1\), choose
\(\psi_k\in\mathcal{B}_p\setminus U\) such that
\(\deg(\psi_k)=d-1\).

The required functions admit a simple monomial choice. For every $k$,
set
$\psi_k(y')=\prod_{i=0}^{d-2}y_i$.
Choose $\phi_0,\ldots,\phi_{q-1}$ as any $q$ distinct ANF monomials,
different from $\psi_k$, whose degrees lie between $2$ and $d-1$.
Such a choice is possible because
$q<\sum_{i=2}^{d-1}\binom{p}{i}$.
Then
$1,y_0,\ldots,y_{p-1},\phi_0,\ldots$, $\phi_{q-1},\psi_k$
are distinct ANF monomials and hence are linearly independent.
Therefore, $\psi_k\notin U$, $\deg(\phi_i)\leq d-1$, and
$\deg(\psi_k)=d-1$, as required. Define

\begin{equation}
\label{Ry}
\begin{aligned}
&R
=
\left(
a_{kr},a'_{kr},b_{kr},b'_{kr}
\right)_{\substack{
0\le k\le 2^q-1\\
0\leq r\leq 2^p-1
}},
\end{aligned}
\end{equation}
where $\deg(A_k)
\le d,
\deg(B_k)\le d$ and
$a_{kr}\oplus a'_{kr}
=\psi_k(v_r),
b_{kr}\oplus b'_{kr}
=\psi_k(v_r)\oplus 1$ for $0\leq k\leq 2^q-1$, $0\leq r\leq 2^p-1$.
The degree conditions on the offset functions are always
feasible. Indeed, for every \(k\) and \(r\), take
\(a_{kr}=b_{kr}=0\), and consequently set
\(a'_{kr}=\psi_k(v_r)\) and
\(b'_{kr}=\psi_k(v_r)\oplus1\). With this choice, $A_k=B_k=0$, so the required degree bounds
are satisfied, while both offset-difference conditions hold
identically.

\begin{Theorem}\label{The-RiposilonGMM}
Under the preceding choices, $\Omega_R^{\mathcal{F}_0}$ is a
maximum-cardinality family of $2^{q+1}$ disjoint spectra
$(q+1)$-plateaued $\mathrm{GMM}$ functions without nonzero linear
structures in $m+p+1$ variables, where $R$ and
$\mathcal{F}_0$ are defined in \eqref{Ry} and
\eqref{Set-Fckr}, respectively. Every member  has algebraic degree exactly $d$.
\end{Theorem}

\begin{proof}
By the assumed linear independence of
$1,y_0,\ldots,y_{p-1},
\phi_0,\ldots,\phi_{q-1}$
and the choice \(\psi_k\notin U\) for every \(k\), all the
assumptions of Theorem~\ref{The-Set-withoutlinear} are
satisfied.
Theorem~\ref{The-Set-withoutlinear} shows that
\(\Omega_{R}^{\mathcal{F}_0}\) is a
maximum-cardinality family of \(2^{q+1}\) disjoint spectra
\((q+1)\)-plateaued $\mathrm{GMM}$ functions without nonzero linear structures. 

It remains to determine the algebraic degrees. With the
preceding choice of the offsets in \eqref{Ry}, the coordinate forms
in \eqref{Eq-GMMREP} are
\begin{equation*}
\begin{aligned}
    &\widetilde{g}_{k}(x,y',y_p)
=(y',h(y')\oplus k)\cdot x\oplus A_k(y')\oplus\psi_k(y')y_p,\\&
\widetilde{g}'_{k}(x,y',y_p)
=(y',h(y')\oplus k)\cdot x
\oplus B_k(y')\oplus\psi_k(y')y_p\oplus y_p.
\end{aligned}
\end{equation*}
Since \(\deg(\phi_i)\leq d-1\) for every \(i\), the
inner-product term
\((y',h(y')\oplus k)\cdot x\) has algebraic degree at most
\(d\). By assumption, 
 $\deg(A_k)\leq d$ and $\deg(B_k)\leq d$.  Hence all terms in the two coordinate forms, except possibly
\(\psi_k(y')y_p\), have algebraic degree at most \(d\).
Moreover, \(\deg(\psi_k)=d-1\), and hence
\(\psi_k(y')y_p\) contains a monomial of degree \(d\).
This monomial involves \(y_p\), whereas the inner-product
term and the offset functions $A_k(y')$ and $B_k(y')$
do not involve \(y_p\), so it cannot be cancelled by any of
these terms. The additional term \(y_p\) in the
second   coordinate form is linear and therefore cannot
cancel a degree-$d$ monomial either. Consequently,
$\deg(\widetilde{g}_{k})
=
\deg(\widetilde{g}'_{k})
=
d$.

Finally, the transformations \(\mathcal{T}_k\) are invertible and
linear, and algebraic degree is invariant under invertible
linear transformations. Therefore,
\(\deg(g_{k})=\deg(g'_{k})=d\) for every
\(0 \leq k \leq 2^q-1\).
\end{proof}
When \(d=p+1\), the preceding degree restrictions on the
functions \(\phi_i\) and on the offset functions
$A_k(y')$ and \(B_{k}(y')\) become
automatic, since they are all Boolean functions in \(p\)
variables. Moreover,
\(\sum_{i=2}^{d-1}\binom{p}{i}=2^p-p-1\), so the parameter
range reduces precisely to the original condition
\(0\leq q<2^p-p-1\). Therefore, under the standing
requirement \(\psi_k\notin U\), the only additional degree
condition needed to attain the optimal algebraic degree
\(p+1\) is \(\deg(\psi_k)=p\). For the explicit monomial choice given above,
this condition is already satisfied by
$\psi_k(y')=\prod_{i=0}^{p-1}y_i$.
The following lemma shows, more generally, that a degree-\(p\)
function outside \(U\) exists for every proper subspace
\(U\subsetneq\mathcal{B}_p\), independently of the preceding
monomial realization.

\begin{Lemma}
\label{Lem-max-degree-auxiliary}
Let $U$ be a proper linear subspace of $\mathcal B_p$. Then there exists a Boolean function
$\psi\in\mathcal B_p\setminus U$
such that
$\deg(\psi)=p$.
\end{Lemma}

\begin{proof}
Let
$\mathcal{V}_{p-1}
=
\{f\in\mathcal B_p\,|\,\deg(f)\le p-1\}$.
Since the monomial
$M_p(y')=y_0y_1\cdots y_{p-1}$
is the unique monomial of degree $p$ in $p$ variables, the set of
all Boolean functions of degree exactly $p$ is the affine coset
$M_p\oplus\mathcal V_{p-1}$.

Suppose, to the contrary, that every Boolean function of degree
$p$ belongs to $U$. Then, for every
$r\in\mathcal V_{p-1}$,
both
$M_p$ and
$M_p\oplus r$
belong to $U$. Since $U$ is a linear subspace,
their sum $r$ also belongs to $U$. Hence
$\mathcal V_{p-1}\subseteq U$.
Together with $M_p\in U$, this implies
$\mathcal B_p
=
\mathcal V_{p-1}\oplus\langle M_p\rangle
\subseteq U$,
contradicting the assumption that $U$ is proper.
Therefore at least one degree $p$ Boolean function lies outside
$U$.
\end{proof}

Under the standing condition \(0\leq q<2^p-p-1\), the
linear independence of
\(1,y_0,\ldots,y_{p-1},\phi_0,\\\ldots,\phi_{q-1}\)
gives \(\dim(U)=p+q+1<2^p=\dim(\mathcal{B}_p)\). Hence \(U\)
is a proper subspace of \(\mathcal{B}_p\). By Lemma~\ref{Lem-max-degree-auxiliary}, for
every \(k\), the function \(\psi_k\) can therefore be chosen
outside \(U\) with \(\deg(\psi_k)=p\). We now specialize the
preceding result to this choice, which yields functions of
the optimal algebraic degree \(p+1\).
By taking \(d=p+1\) in Theorem~\ref{The-RiposilonGMM}, we obtain the optimal-degree
case recorded in the following corollary.
\begin{Corollary}\label{Cor1}
Under the hypotheses of Theorem \ref{The-Set-withoutlinear}, if additionally  \(\deg(\psi_k)=p\) for every
\(0\leq k\leq2^q-1\), then all functions in
\(\Omega_{R}^{\mathcal {F}_0}\) constructed in Theorem~\ref{The-Set-withoutlinear}
attain the optimal algebraic degree:
$\deg(g_{k})=\deg(g'_{k})=p+1$
for every \(0\leq k\leq2^q-1\).
\end{Corollary}

The following example illustrates an optimal-degree family
of disjoint spectra plateaued functions without nonzero
linear structures.

\begin{Example}\label{Ex-optimal-family}
We give a concrete realization of the construction for $q=1$.
Take $p=4$ and $m=p+q=5$, and write
$x=(x_0,x_1,x_2,x_3,x_4)$ and
$y=(y',y_4)=(y_0,y_1,y_2,y_3,y_4)$.
The parameter condition is satisfied since
$1<2^4-4-1=11$.
Take \(\{v_r\}_{r=0}^{15}=\mathbb F_2^4\),
with \(v_0,\ldots,v_{15}\) in lexicographic order.

Choose $\phi_0(y')=y_0y_1$, where
$y'=(y_0,y_1,y_2,y_3)$. Then
$U=\operatorname{span}
\{1,y_0,y_1,y_2,y_3,y_0y_1\}$
is a proper subspace of $\mathcal B_4$.
Choose
$\psi_0(y')=y_0y_1y_2y_3$ and
$\psi_1(y')
=y_0y_1y_2y_3\oplus y_0y_2\oplus y_1y_3$.
Clearly, $\psi_0,\psi_1\notin U$ and
$\deg(\psi_0)=\deg(\psi_1)=4=p$.

Since $h(y')=y_0y_1$, put
$c_{kr}=(v_r,h(v_r)\oplus k)$ for
$k\in\mathbb F_2$ and $0\leq r\leq15$.
Writing $v_r=(v_{r0},v_{r1},v_{r2},v_{r3})$, we have
\[
f_{kr}(x)=c_{kr}\cdot x
=
v_{r0}x_0\oplus v_{r1}x_1\oplus v_{r2}x_2
\oplus v_{r3}x_3
\oplus (v_{r0}v_{r1}\oplus k)x_4.
\]
Thus,
$\mathcal {F}_0
=\{f_{kr}\mid k\in\mathbb F_2,\ 0\leq r\leq15\}$.

Moreover,
$\mathcal D_k=\{c_{kr}\mid 0\leq r\leq15\}$, and the two
sets are explicitly given by
\[
\begin{aligned}
\mathcal D_0
  &=\{0,1,2,4,5,6,8,9,10,12,13,14,19,23,27,31\},\\
\mathcal D_1
  &=\{3,7,11,15,16,17,18,20,21,22,24,25,26,28,29,30\}.
\end{aligned}
\]
Hence $\mathcal D_0\cap\mathcal D_1=\varnothing$ and
$\mathcal D_0\cup\mathcal D_1=\mathbb F_2^5$.

For each $k$ satisfying $0\leq k\leq 1$, choose
$S_k=S=\mathbb F_2^4\times\{0\}$ and take the linear
isomorphism
$\mathcal L_k:\mathbb F_2^4\rightarrow S$ given by
$\mathcal L_k(y')=(y',0)$.
Let $\gamma_0=(0,0,0,0,1)$ and
$\gamma_1=(1,0,0,0,1)$.
Then $\gamma_0,\gamma_1\notin S$ and
$\gamma_0\neq\gamma_1$.
For every $0\leq r\leq15$, set
$s_{kr}=\mathcal L_k(v_r)=(v_r,0)$.
Consequently,
$s'_{0r}=s_{0r}\oplus\gamma_0=(v_{r0},v_{r1},v_{r2},v_{r3},1)$ and
$s'_{1r}=s_{1r}\oplus\gamma_1
=(v_{r0}\oplus1,v_{r1},v_{r2},v_{r3},1)$.

Finally, for every $k\in\mathbb F_2$ and
$0\leq r\leq15$, choose
$a_{kr}=b_{kr}=0$,
$a'_{kr}=\psi_k(v_r)$, and
$b'_{kr}=\psi_k(v_r)\oplus1$.
Then
$a_{kr}\oplus a'_{kr}=\psi_k(v_r)$ and
$b_{kr}\oplus b'_{kr}=\psi_k(v_r)\oplus1$,
so these parameters satisfy the defining conditions on $R$
in~\eqref{Set-R(psi)}.
Moreover, the associated offset functions satisfy
$A_k=B_k=0$ for each $k\in\mathbb F_2$.

Set
\[
\begin{aligned}
P_0(x,y)
=&
y_0y_1y_2y_3y_4
\oplus x_4y_0y_1
\oplus x_0y_0
\oplus x_1y_1
\oplus x_2y_2
\oplus x_3y_3,\\
P_1(x,y)
=&
y_0y_1y_2y_3y_4
\oplus y_1y_2y_3y_4
\oplus y_0y_2y_4
\oplus y_1y_3y_4
\oplus x_4y_0y_1
\oplus\\& x_4y_1y_4
\oplus y_2y_4
\oplus x_0y_0
\oplus x_0y_4
\oplus x_1y_1
\oplus x_2y_2\oplus x_3y_3
\oplus x_4.
\end{aligned}
\]
Substituting the above choices into~\eqref{Eq-g_k} gives
\begin{equation*}
\begin{aligned}
&g_0(x,y)=P_0(x,y),\qquad
g'_0(x,y)=P_0(x,y)\oplus y_4,\\
&g_1(x,y)=P_1(x,y),\qquad
g'_1(x,y)=P_1(x,y)\oplus y_4.
\end{aligned}
\end{equation*}

Therefore, $\{g_0,g'_0,g_1,g'_1\}$ is a
maximum-cardinality family of four disjoint spectra
$2$-plateaued functions without nonzero linear structures in $10$ variables. Moreover,
$\deg(g_0)=\deg(g'_0)=\deg(g_1)=\deg(g'_1)=5$,
so all four members attain the optimal algebraic degree
simultaneously.

Finally, a direct calculation of the Walsh-support self-sums
gives
\[
\begin{aligned}
\bigl|
\operatorname{supp}(W_{g_0})
\oplus
\operatorname{supp}(W_{g_0})
\bigr|
&=
\bigl|
\operatorname{supp}(W_{g'_0})
\oplus
\operatorname{supp}(W_{g'_0})
\bigr|
=688,\\
\bigl|
\operatorname{supp}(W_{g_1})
\oplus
\operatorname{supp}(W_{g_1})
\bigr|
&=
\bigl|
\operatorname{supp}(W_{g'_1})
\oplus
\operatorname{supp}(W_{g'_1})
\bigr|
=880.
\end{aligned}
\]
Thus, the Walsh-support self-sum cardinality is not constant across
the family. It follows from Proposition~\ref{Pro-5} that this family
is not memberwise EA-equivalent to any family obtained from
\cite[Theorem~4.4]{disjoint2}.
\end{Example}
\section{Families of disjoint spectra plateaued functions
without nonzero linear structures based on partially linear
functions with bent components}\label{sec-partially-bent}
The explicit construction in
Section~\ref{sec-plateauedwithoutlinear} was obtained by taking
$n=m=p+q$ and choosing the constituent functions $f_{kr}$ to be
linear. We now return to the general framework of
Section~\ref{sec-0} and assume that $n>m$ and $n-m$ is even.
The partition
$\{\mathcal D_k\}_{k\in\mathbb F_2^q}$ defined
in~\eqref{Eq-Dk}, the vectors $c_{kr}$ defined
in~\eqref{ckr}, the selector maps $\sigma_k$, and the offset
family $R$ defined in~\eqref{Set-R(psi)} are retained unchanged.
The constituent functions are now taken to be 
partially linear functions with bent components.
\subsection{Construction and structural properties}

Let $p\geq 2$, $0\leq q<2^p-p-1$, and set $m=p+q$. Let
$n>m$ be such that $n-m$ is even, and write
 $x=(x',x'')\in
 \mathbb{F}_2^{\,n-m}\times\mathbb{F}_2^m$.
For every $0\leq k\leq 2^q-1$ and
$0\leq r\leq 2^p-1$, choose an arbitrary bent function
 $f_{kr}^{*}\in\mathcal{B}_{n-m}$,
and define
\[
 f_{kr}(x',x'')
 =
 f_{kr}^{*}(x')\oplus c_{kr}\cdot x'',
\]
where $c_{kr}\in \mathcal{D}_k$, with $\mathcal{D}_k$  defined in \eqref{Eq-Dk}.
Accordingly, let
\begin{equation}
 \mathcal{F}_1
 =
 \left\{
 f_{kr}\,|\,
 f_{kr}(x',x'')
 =
 f_{kr}^{*}(x')\oplus c_{kr}\cdot x'',
 0\leq k\leq 2^q-1,
 0\leq r\leq 2^p-1
 \right\}.
 \label{Ffckr}
\end{equation}
The sets $\mathcal{D}_k$ form a partition of
$\mathbb{F}_2^m$. Hence all the vectors $c_{kr}$, with
$0\leq k\leq 2^q-1$ and $0\leq r\leq 2^p-1$, are distinct.
Therefore, $\mathcal{F}_1$ is a family of  disjoint spectra $m$-plateaued functions.

\begin{Theorem}
\label{thm:bent-extension}
 The family $\Omega_{R}^{\mathcal{F}_1}$ defined in
\eqref{Set-Omega} is a maximum-cardinality family of disjoint spectra
$(q+1)$-plateaued $\mathrm{GMM}$ functions without nonzero linear
structures in $n+p+1$ variables, where $\mathcal{F}_1$ is given by \eqref{Ffckr} and
$R$ is given by \eqref{Set-R(psi)}.
\end{Theorem}

\begin{proof}
Since $\mathcal{F}_1$ is a family of disjoint spectra $m$-plateaued functions,
Proposition~\ref{The-disjoint}  shows that
$\Omega_R^{\mathcal{F}_1}$ is a maximum-cardinality family of
$2^{q+1}$ disjoint spectra $(q+1)$-plateaued functions  in $n+p+1$ variables.

Every member of $\Omega_R^{\mathcal{F}_1}$ also admits a $\mathrm{GMM}$
representation. Define
 $\pi_k(x',y)=c_{k,\sigma_k(y)}\in\mathbb{F}_2^m$,
together with
\[
h_k(x',y)
 =
 f_{k,\sigma_k(y)}^{*}(x')\oplus a_k(y),
 \qquad
h_k'(x',y)
 =
f_{k,\sigma_k(y)}^{*}(x')\oplus b_k(y).
\]
Then
\begin{equation*}
    \begin{aligned}
         &g_k(x',x'',y)
 =
 x''\cdot\pi_k(x',y)\oplus h_k(x',y),\\
 &g_k'(x',x'',y)
 =
 x''\cdot\pi_k(x',y)\oplus h_k'(x',y).
    \end{aligned}
\end{equation*}
Since the vectors $c_{kr}$ are distinct and
$\sigma_k^{-1}(r)=\{s_{kr},s'_{kr}\}$, we have
$\pi_k^{-1}(c_{kr})
 =\mathbb F_2^{\,n-m}\times\{s_{kr},s'_{kr}\}$.
Hence, $\pi_k$ is exactly $2^{n-m+1}$-to-one onto
$\mathcal D_k$.

Consequently, both $g_k$ and $g_k'$ admit $\mathrm{GMM}$ representations for $0\leq k\leq 2^q-1$.

We next exclude nonzero linear structures. Fix
$0\leq k\leq 2^q-1$ and consider a direction
 $(\alpha',\alpha'',\beta)
 \in
 \mathbb{F}_2^{\,n-m}\times
 \mathbb{F}_2^m\times
 \mathbb{F}_2^{p+1}$.
 \begin{equation*}
 \begin{aligned}
        D_{(\alpha',\alpha'',\beta)}g_k=&f_{k,\sigma_k(y)}(x',x'')\oplus f_{k,\sigma_k(y\oplus\beta)}(x'\oplus\alpha',x''\oplus\alpha'')\oplus a_{k}(y)\oplus a_{k}(y\oplus\beta)\\=&f_{k,\sigma_k(y)}^{*}(x')
\oplus
f_{k,\sigma_k(y\oplus\beta)}^{*}(x'\oplus\alpha')\oplus
\bigl(
c_{k,\sigma_k(y)}
\oplus
c_{k,\sigma_k(y\oplus\beta)}
\bigr)\cdot x''\oplus\\&
c_{k,\sigma_k(y\oplus\beta)}\cdot\alpha''
\oplus a_k(y)
\oplus a_k(y\oplus\beta).
 \end{aligned}
 \end{equation*}

As in the proof of Theorem~1, the derivatives in all nonzero
directions of the form $(\boldsymbol{0}_{n-m},\alpha'',\beta)$ are nonconstant.
Indeed, when $\beta\in\{\boldsymbol{0}_{p+1},\gamma_k\}$, the two terms involving the same bent component cancel, and the corresponding argument in
Theorem~\ref{The-Set-withoutlinear} applies directly. When
$\beta\notin\{\boldsymbol{0}_{p+1},\gamma_k\}$, the derivative contains the
nonzero linear term
 $\bigl(c_{k,\sigma_k(y)}
 \oplus c_{k,\sigma_k(y\oplus\beta)}\bigr)\cdot x''$,
and is therefore nonconstant. Hence, it remains only to consider
directions satisfying $\alpha'\neq \boldsymbol{0}_{n-m}$.

First, suppose that
$\beta\notin\{\boldsymbol{0}_{p+1},\gamma_k\}$. As in Case~1 of the proof of
Theorem~\ref{The-Set-withoutlinear}, we have
 $\sigma_k(y)\neq\sigma_k(y\oplus\beta)$
for every $y\in\mathbb{F}_2^{p+1}$. Since the vectors
$c_{kr}$ are distinct, the coefficient
 $c_{k,\sigma_k(y)}
 \oplus c_{k,\sigma_k(y\oplus\beta)}$
of $x''$ in the derivative is nonzero. Thus, the derivatives of
both $g_k$ and $g_k'$ are nonconstant.

It remains to consider
$\beta\in\{\boldsymbol{0}_{p+1},\gamma_k\}$. In this case,
 $\sigma_k(y\oplus\beta)=\sigma_k(y)$.
Then
\begin{align*}
 D_{(\alpha',\alpha'',\beta)}g_k(x',x'',y)
 &=
 D_{\alpha'}f_{k,\sigma_k(y)}^{*}(x')
 \oplus c_{k,\sigma_k(y)}\cdot\alpha''
 \oplus a_k(y)\oplus a_k(y\oplus\beta),\\
 D_{(\alpha',\alpha'',\beta)}g_k'(x',x'',y)
 &=
 D_{\alpha'}f_{k,\sigma_k(y)}^{*}(x')
 \oplus c_{k,\sigma_k(y)}\cdot\alpha''
 \oplus b_k(y)\oplus b_k(y\oplus\beta).
\end{align*}
Since $\alpha'\neq \boldsymbol{0}_{n-m}$ and $f_{kr}^{*}$ is bent,
$D_{\alpha'}f_{kr}^{*}$ is balanced, and hence nonconstant.
All the remaining terms in the above expressions are independent
of $x'$. Therefore, neither derivative can be constant.

Combining the preceding cases, neither $g_k$ nor $g_k'$ admits
a nonzero linear structure. Since $k$ was arbitrary, no member
of $\Omega_R^{\mathcal{F}_1}$ admits a nonzero linear structure.
\end{proof}

The family obtained in Theorem~\ref{thm:bent-extension} exhibits the same structural 
distinction from the spectral construction in \cite[Theorem~4.4]{disjoint2}.
We first determine the Walsh supports of its members.

Let $\omega'\in\mathbb{F}_2^{\,n-m}$ and
$\omega''\in\mathbb{F}_2^m$. Since $f_{kr}^{*}$ is bent,
$W_{f_{kr}^{*}}(\omega')
=
2^{(n-m)/2}
(-1)^{\widetilde{f_{kr}^{*}}(\omega')}$, where $\widetilde{f_{kr}^{*}}$ is the dual function of $f_{kr}^{*}$ for $0\leq k\leq 2^q-1, 0\leq r\leq 2^p-1$.  It follows that
\begin{align*}
W_{f_{kr}}(\omega',\omega'')
=
W_{f_{kr}^{*}}(\omega')
\sum_{x''\in\mathbb{F}_2^m}
(-1)^{(c_{kr}\oplus\omega'')\cdot x''}
=
\begin{cases}
2^{(n+m)/2}
(-1)^{\widetilde{f_{kr}^{*}}(\omega')},
&\omega''=c_{kr},\\
0,&\omega''\neq c_{kr}.
\end{cases}
\end{align*}
Consequently,
$\operatorname{supp}(W_{f_{kr}})
=
\mathbb{F}_2^{\,n-m}\times\{c_{kr}\}$.
By the Walsh-support description in~\eqref{Suppgg'} and the choice of the
offset parameters in \eqref{Set-R(psi)}, we obtain
\begin{align*}
\operatorname{supp}(W_{g_k})
&=
\bigcup_{r=0}^{2^p-1}
\left\{
(\omega',c_{kr},\omega''')\
\middle|\
\omega'\in\mathbb{F}_2^{\,n-m},\
\gamma_k\cdot\omega'''=\psi_k(v_r)
\right\},\\
\operatorname{supp}(W_{g_k'})
&=
\bigcup_{r=0}^{2^p-1}
\left\{
(\omega',c_{kr},\omega''')\
\middle|\
\omega'\in\mathbb{F}_2^{\,n-m},
\gamma_k\cdot\omega'''=\psi_k(v_r)\oplus1
\right\},
\end{align*}
where $\omega'''\in\F_2^{p+1}$.
The bent components affect only the signs of the nonzero Walsh
coefficients; at the level of Walsh supports, they merely add the
unrestricted factor $\mathbb{F}_2^{\,n-m}$. Thus, on the remaining
$m+p+1$ spectral coordinates, the support structure is exactly
the one treated in Proposition~4. Therefore, if $\gamma_k\neq\gamma_{k'}$ for some $k\neq k'$,
Proposition~\ref{Prop-direct-separation} implies that no Walsh
support associated with $\{g_k,g_k'\}$ is a translation of one
associated with $\{g_{k'},g_{k'}'\}$. In contrast, every family
obtained from \cite[Theorem~4.4]{disjoint2} satisfies the
common-translation condition, which is preserved under a common
EA transformation. Consequently, $\Omega_R^{\mathcal{F}_1}$ is
not commonly EA-equivalent to any family obtained from that
construction.

The construction based on partially linear functions with bent
components therefore also yields maximum-cardinality families
of disjoint spectra plateaued functions without nonzero
linear structures. We next determine
the attainable algebraic degrees of these families. Unlike the
construction in Theorem \ref{The-Set-withoutlinear}, their maximum attainable
degree is generally next-to-optimal rather than optimal.

\subsection{Algebraic-degree control}
We first determine the attainable algebraic degrees of the
concatenation of the bent components. Recall that \(\{v_r\}_{r=0}^{2^p-1}=\mathbb F_2^p\), and that \(\rho(y')\) denotes the unique index
satisfying \(v_{\rho(y')}=y'\). For every \(k\), define
\begin{equation}\label{PHI}
    \Phi_k(x',y')=f_{k,\rho(y')}^{*}(x').
\end{equation}
Writing \(v_r=(v_{r0},\ldots,v_{r,p-1})\), define the indicator
of \(v_r\) by
$\delta_r(y')
 =\prod_{i=0}^{p-1}
  (y_i\oplus v_{ri}\oplus1)$.
Thus, \(\delta_r(y')=1\) if and only if \(y'=v_r\). Consequently,
the concatenated function admits the indicator representation
\[
\Phi_k(x',y')
 =\bigoplus_{r=0}^{2^p-1}
 \delta_r(y')f_{kr}^{*}(x').
\]
Indeed, for every fixed \(y'\), the only nonzero term in the above
sum is the one indexed by \(r=\rho(y')\).

\begin{Lemma}\label{Phhhi}

Let \(n-m\) be positive and even. The maximum possible algebraic
degree of \(\Phi_k\) in \eqref{PHI} is \(p+(n-m)/2\). Moreover, every integer
degree \(2\leq d\leq p+(n-m)/2\) can be attained simultaneously
for all \(k\) by a suitable choice of the bent components
\(f_{kr}^{*}\in\mathcal{B}_{n-m}\).
\end{Lemma}

\begin{proof}
We first establish the upper bound. Suppose that \(n-m\geq4\).
Every \((n-m)\)-variable bent function has degree at most
\((n-m)/2\). Since \(\deg(\delta_r)=p\), the indicator
representation of \(\Phi_k\) gives
\(\deg(\Phi_k)\leq p+(n-m)/2\).

When \(n-m=2\), every bent component has the form
\(x_0x_1\oplus r_{kr}x_0\oplus
r'_{kr}x_1\oplus r''_{kr}\), where $r_{kr}, r'_{kr}, r''_{kr}\in\F_2$. 
Using the indicator representation, we obtain
\begin{equation*}
    \begin{aligned}
        \Phi_k(x_0,x_1,y')
&=\bigoplus_{r=0}^{2^p-1}
  \delta_r(y')f_{kr}^{*}(x_0,x_1)\\
&=\left(\bigoplus_{r=0}^{2^p-1}\delta_r(y')\right)x_0x_1\oplus
 \left(\bigoplus_{r=0}^{2^p-1}
 r_{kr}\delta_r(y')\right)x_0\oplus
 \left(\bigoplus_{r=0}^{2^p-1}
 r_{kr}'\delta_r(y')\right)x_1
 \oplus
 \bigoplus_{r=0}^{2^p-1}
 r_{kr}''\delta_r(y')
 \\&=x_0x_1\oplus
 \left(\bigoplus_{r=0}^{2^p-1}
 r_{kr}\delta_r(y')\right)x_0\oplus
 \left(\bigoplus_{r=0}^{2^p-1}
 r_{kr}'\delta_r(y')\right)x_1
 \oplus
 \bigoplus_{r=0}^{2^p-1}
 r_{kr}''\delta_r(y'),
    \end{aligned}
\end{equation*}
where the coefficient functions have degree at most \(p\).
Therefore,
\(\deg(\Phi_k)\leq p+1=p+(n-m)/2\).

It remains to show that every degree in the stated range is
attainable.

Suppose first that
\(p+2< d\leq p+(n-m)/2\). Choose an index \(r'\) and a bent
function \(f_{kr'}^{*}\) of degree \(d-p\). Let \(M(x')\) be
a monomial of degree \(d-p\) occurring in its ANF. Choose all the
remaining bent components to have degree less than \(d-p\), so
that none of them contains \(M(x')\).

We now consider separately the boundary case \(d-p=2\), in which bent functions of degree less than \(2\) do not exist.
Moreover, since \(d\leq p+(n-m)/2\), we necessarily have
\(n-m\geq4\). Take \(M(x')=x_0x_1\). For \(n-m=4\), choose
\(f_{kr'}^{*}(x')=x_0x_1\oplus x_2x_3\) and
\(f_{kr}^{*}(x')=x_0x_2\oplus x_1x_3\) for every \(r\neq r'\).
These functions are quadratic bent, and \(M(x')\) occurs only in
\(f_{kr'}^{*}\). When \(n-m>4\), the same term
\(\bigoplus_{j=2}^{(n-m)/2-1}x_{2j}x_{2j+1}\) can be added to
every component.

It follows from the indicator representation that the contribution
of \(M(x')\) to \(\Phi_k\) is
\(\delta_{r'}(y')M(x')\). Since
\(\deg(\delta_{r'})=p\), this term has degree
\(p+(d-p)=d\). Moreover, every bent component has degree at most
\(d-p\), so every summand
\(\delta_r(y')f_{kr}^{*}(x')\) has degree at most
\(p+(d-p)=d\). Therefore, \(\deg(\Phi_k)=d\).

Suppose next that \(2\leq d\leq p+1\). Choose any bent function
\(\theta_k(x')\) satisfying \(\deg(\theta_k)\leq d\), any
nonconstant affine function \(\lambda_k(x')\), and any Boolean
function \(Q_k\in\mathcal{B}_p\) of degree \(d-1\). For every
\(0\leq r\leq2^p-1\), define
\begin{equation}\label{1234}
    f_{kr}^{*}(x')
 =\theta_k(x')\oplus Q_k(v_r)\lambda_k(x').
\end{equation}
Since \(Q_k(v_r)\) is a constant in \(\mathbb F_2\), every
\(f_{kr}^{*}\) is either \(\theta_k\) or
\(\theta_k\oplus\lambda_k\). Adding an affine function preserves
bentness, so all the resulting components are bent.

Using \(v_{\rho(y')}=y'\), the corresponding concatenated function
is
\begin{equation}\label{456}
    \Phi_k(x',y')
 =f_{k,\rho(y')}^{*}(x')
 =\theta_k(x')
   \oplus Q_k(v_{\rho(y')})\lambda_k(x')
 =\theta_k(x')
   \oplus Q_k(y')\lambda_k(x').
\end{equation}
Since \(\deg(Q_k)=d-1\) and \(\deg(\lambda_k)=1\), the second term
has degree \(d\). Its degree-\(d\) monomials involve variables from
both \(x'\) and \(y'\), whereas \(\theta_k\) is independent of
\(y'\); hence, they cannot be cancelled. Since
\(\deg(\theta_k)\leq d\), it follows that
\(\deg(\Phi_k)=d\).

The above choices can be made independently for every \(k\), which
completes the proof.
\end{proof}

Let \(p\geq2\), \(n>m=p+q\), and let \(n-m\) be even. Let
\(3\leq d\leq p+(n-m)/2\), and suppose that
$0\leq q<
\sum_{i=2}^{\min\{d-1,p\}}\binom{p}{i}$.
In Subsection~\ref{subse-partitionDk}, choose Boolean functions
\(\phi_0,\ldots,\phi_{q-1}\in\mathcal{B}_p\) such that
\(1,y_0,\ldots,y_{p-1},\phi_0,\ldots,\phi_{q-1}\) are linearly
independent and
\(\deg(\phi_i)\leq\min\{d-1,p\}\) for every \(0\leq i<q\).
Let \(U\) be the subspace defined in~\eqref{spanU}. For every
\(0\leq k<2^q\), choose
\(\psi_k\in\mathcal{B}_p\setminus U\) satisfying
\(\deg(\psi_k)\leq\min\{d-1,p\}\). Choose the offset parameters
$R$ defined in \eqref{Ry}, where $\deg(A_k)
\le d,
\deg(B_k)\le d$ and
$a_{kr}\oplus a'_{kr}
=\psi_k(v_r),
b_{kr}\oplus b'_{kr}
=\psi_k(v_r)\oplus 1$ for $0\leq k\leq 2^q-1$, $0\leq r\leq 2^p-1$.

As in Subsection~\ref{Subsec-optimal-degree-degenerate}, these choices are always feasible under the
preceding bound on \(q\). For example, one may take
$\psi_k(y')
 =\prod_{i=0}^{\min\{d-1,p\}-1}y_i$
for every \(k\), and choose
\(\phi_0,\ldots,\phi_{q-1}\) as any \(q\) distinct ANF monomials,
different from \(\psi_k\), whose degrees lie between \(2\) and
\(\min\{d-1,p\}\). Then
\(1,y_0,\ldots,y_{p-1},\phi_0,\ldots,\phi_{q-1},\psi_k\) are
distinct ANF monomials and hence linearly independent. In
particular, \(\psi_k\notin U\). Then  for every \(k\) and \(r\), take
\(a_{kr}=b_{kr}=0\),
\(a'_{kr}=\psi_k(v_r)\) and
\(b'_{kr}=\psi_k(v_r)\oplus1\).

Finally, for every \(0\leq k<2^q\), choose the bent components
\(f_{kr}^{*}\in\mathcal{B}_{n-m}\), \(0\leq r<2^p\), as follows.
If \(3\leq d\leq p+1\), choose \(f_{kr}^{*}\) such that the
concatenated function
\(\Phi_k(x',y')=f_{k,\rho(y')}^{*}(x')\) has degree \(d\) and its
degree-\(d\) homogeneous part contains a monomial involving a
variable from \(x'\). Such a choice is possible by taking
\(f_{kr}^{*}\) as in~\eqref{1234}, whose concatenation \(\Phi_k\) is given
by~\eqref{456}. If \(p+2\leq d\leq p+(n-m)/2\), choose
\(f_{kr}^{*}\) according to Lemma~\ref{Phhhi} so that
\(\deg(\Phi_k)=d\).

\begin{Theorem}
\label{thm:bent-prescribed-degree}
Under the preceding choices, $\Omega_R^{\mathcal{F}_1}$ is a
maximum-cardinality family of $2^{q+1}$ disjoint spectra
$(q+1)$-plateaued $\mathrm{GMM}$ functions without nonzero
linear structures in $n+p+1$ variables. Every member has
algebraic degree exactly $d$. In particular, the maximum algebraic degree attainable by the
construction based on partially linear functions with bent
components is
\(p+(n-m)/2\).
\end{Theorem}

\begin{proof}
Theorem~\ref{thm:bent-extension} shows that
$\Omega_R^{\mathcal{F}_1}$ is a maximum-cardinality family
of disjoint spectra $(q+1)$-plateaued $\mathrm{GMM}$
functions without nonzero linear structures. It remains to
determine their algebraic degrees.

For every $k$, apply the invertible linear transformation
$\mathcal{T}_k$ defined in \eqref{Eq-isoT}. Since
$v_{\rho(y')}=y'$ and
$\sigma_k(\mathcal{T}_k(y',y_p))=\rho(y')$, we have
$c_{k,\rho(y')}=(y',h(y')\oplus k)$. Moreover,
$a_k(\mathcal{T}_k(y',y_p))
=A_k(y')\oplus\psi_k(y')y_p$ and
$b_k(\mathcal{T}_k(y',y_p))
=B_k(y')\oplus\psi_k(y')y_p\oplus y_p$. By the definition of $\Phi_k$ in \eqref{PHI}, it follows that
\begin{align*}
\widetilde{g}_k(x',x'',y',y_p)
&=
\Phi_k(x',y')
\oplus (y',h(y')\oplus k)\cdot x''
\oplus A_k(y')
\oplus\psi_k(y')y_p,\\
\widetilde{g}_k'(x',x'',y',y_p)
&=
\Phi_k(x',y')
\oplus (y',h(y')\oplus k)\cdot x''
\oplus B_k(y')
\oplus\psi_k(y')y_p
\oplus y_p.
\end{align*}

Suppose first that \(3\leq d\leq p+1\). By the preceding choices,
\(\deg(\phi_i)\leq d-1\), so the inner-product term
\((y',h(y')\oplus k)\cdot x''\) has degree at most \(d\).
Moreover, \(\deg(A_k),\deg(B_k)\leq d\), and
\(\deg(\psi_k)\leq d-1\). Hence, all the terms in the coordinate
forms other than \(\Phi_k\) have degree at most \(d\). Moreover,
these terms are independent of \(x'\), whereas the degree-\(d\)
homogeneous part of \(\Phi_k\) contains a monomial involving a
variable from \(x'\). Therefore, this monomial cannot be cancelled,
and every coordinate function has algebraic degree exactly \(d\).

Suppose next that
\(p+1<d\leq p+(n-m)/2\). In this range, the preceding degree
conditions on \(\phi_i,\psi_k,A_k\), and \(B_k\) are automatic,
since they are all Boolean functions in the \(p\) variables
\(y'\). In particular, the inner-product term and
\(\psi_k(y')y_p\) have degree at most \(p+1<d\), while
\(A_k\), \(B_k\), and \(y_p\) have degree strictly less than \(d\).
 Lemma~\ref{Phhhi} ensures that
\(\deg(\Phi_k)=d\), while all the remaining terms in the coordinate
forms have degree strictly less than \(d\). Hence, they cannot
cancel the degree-\(d\) part of \(\Phi_k\), and therefore
\(\deg(\widetilde{g}_k)=\deg(\widetilde{g}'_k)=d\) also in this
case.

Combining the two cases, and since \(\mathcal{T}_k\) is
invertible and linear, algebraic degree is preserved. Consequently,
\(\deg(g_k)=\deg(g'_k)=d\) for every \(0\leq k\leq2^q-1\).

Finally, Lemma~\ref{Phhhi} gives
\(\deg(\Phi_k)\leq p+(n-m)/2\), while all the remaining terms in
the coordinate forms have degree at most \(p+1\). Since \(n-m\)
is positive and even, \(p+1\leq p+(n-m)/2\). Thus, the degree of
a function obtained by the present construction cannot exceed
\(p+(n-m)/2\). Taking \(d=p+(n-m)/2\) attains this bound, which is
therefore the maximum algebraic degree of the construction.  On the other hand, the optimal algebraic-degree bound for a
$(q+1)$-plateaued Boolean function in $n+p+1$ variables is
\[
\frac{(n+p+1)-(q+1)}{2}+1
=p+\frac{n-m}{2}+1.
\]
Hence, the maximum degree attained by the present construction
is next-to-optimal.
\end{proof}
  The following example gives a next-to-optimal-degree family in
$10$ variables obtained from four distinct quadratic bent
functions.

\begin{Example}
We give a concrete realization of the construction for $q=1$.
Take $p=3$, $m=p+q=4$, and $n=6$, and write
$x'=(x_0,x_1)$, $x''=(x_2,x_3,x_4,x_5)$, and
$y=(y',y_3)$, where $y'=(y_0,y_1,y_2)$.
The parameter condition is satisfied since
$1<2^3-3-1=4$, and $n-m=2$ is positive and even.
Take \(\{v_r\}_{r=0}^{7}=\mathbb F_2^3\),
with \(v_0,\ldots,v_7\) in lexicographic order.

Choose $\phi_0(y')=y_0y_1$. Then
$U=\operatorname{span}\{1,y_0,y_1,y_2,y_0y_1\}$
is a proper subspace of $\mathcal B_3$.
Choose $\psi_0(y')=y_0y_2$ and
$\psi_1(y')=y_0y_1y_2$.
Clearly, $\psi_0,\psi_1\notin U$.

For every $k\in\mathbb F_2$, choose the bent components as
\[
f_{kr}^{*}(x')=
\begin{cases}
x_0x_1, & r\in\{0,1,2\},\\
x_0x_1\oplus1, & r=3,\\
x_0x_1\oplus x_1, & r\in\{4,5,6\},\\
x_0x_1\oplus x_0\oplus x_1\oplus1, & r=7.
\end{cases}
\]
Thus, the bent components take exactly four distinct values, each of
which is a quadratic bent function obtained by adding an affine
function to $x_0x_1$. By the lexicographic enumeration, their
concatenation is
\[
\Phi_k(x',y')
=f_{k,\rho(y')}^{*}(x')
= y_0y_1y_2x_0\oplus x_0x_1
 \oplus y_0x_1\oplus y_1y_2.
\]

Since $h(y')=y_0y_1$, put
$c_{kr}=(v_r,h(v_r)\oplus k)$ for
$k\in\mathbb F_2$ and $0\leq r\leq7$, and define
$f_{kr}(x',x'')=f_{kr}^{*}(x')\oplus c_{kr}\cdot x''$.
Writing $v_r=(v_{r0},v_{r1},v_{r2})$, we have
\[
f_{kr}(x',x'')
=x_0x_1
 \oplus v_{r0}v_{r1}v_{r2}x_0
 \oplus v_{r0}x_1
 \oplus v_{r1}v_{r2}
\oplus v_{r0}x_2
 \oplus v_{r1}x_3
 \oplus v_{r2}x_4
 \oplus (v_{r0}v_{r1}\oplus k)x_5.
\]
Thus,
$\mathcal F_1=\{f_{kr}\mid k\in\mathbb F_2,\ 0\leq r\leq7\}$.

Moreover, $\mathcal{D}_k=\{c_{kr}\mid 0\leq r\leq 7\}$, and the two sets are
$\mathcal D_0=\{0,1,2,4,5,6,11,15\}, 
\mathcal D_1=\{3,7,8,9,10,12,13,14\}.$
Hence, $\mathcal{D}_0\cap\mathcal{D}_1=\varnothing$ and
$\mathcal{D}_0\cup \mathcal{D}_1=\mathbb F_2^4$.

For each $k$ satisfying $0\leq k\leq 1$, choose
$S_k=S=\mathbb F_2^3\times\{0\}$ and take the linear isomorphism
$\mathcal L_k:\mathbb F_2^3\to S$ given by
$\mathcal L_k(y')=(y',0)$.
Let $\gamma_0=\gamma_1=(0,0,0,1)$.
Then $\gamma_0,\gamma_1\notin S$.
For every $0\leq r\leq7$, set
$s_{kr}=\mathcal L_k(v_r)=(v_r,0)$.
Consequently,
$s'_{0r}=s'_{1r}=(v_{r0},v_{r1},v_{r2},1)$.

Finally, for every $k\in\mathbb F_2$ and $0\leq r\leq7$, choose
$a_{kr}=b_{kr}=0$,
$a'_{kr}=\psi_k(v_r)$, and
$b'_{kr}=\psi_k(v_r)\oplus1$.
Then $a_{kr}\oplus a'_{kr}=\psi_k(v_r)$ and
$b_{kr}\oplus b'_{kr}=\psi_k(v_r)\oplus1$, so these parameters
satisfy the defining conditions on $R$ in \eqref{Set-R(psi)}.
Moreover, the associated offset functions satisfy
$A_k=B_k=0$ for each $k\in\mathbb F_2$.

Set
\begin{equation*}
    \begin{aligned}
        P_0(x',x'',y)
=&x_0y_0y_1y_2
 \oplus x_5y_0y_1
 \oplus y_0y_2y_3
 \oplus x_0x_1
 \oplus x_1y_0
 \oplus y_1y_2
\oplus x_2y_0
 \oplus x_3y_1
 \oplus x_4y_2
\\P_1(x',x'',y)
=&x_0y_0y_1y_2
 \oplus y_0y_1y_2y_3
 \oplus x_5y_0y_1
 \oplus x_0x_1
 \oplus x_1y_0
 \oplus y_1y_2\oplus x_2y_0
 \oplus x_3y_1
 \oplus x_4y_2
 \oplus x_5.
    \end{aligned}
\end{equation*}
Substituting the above choices into the construction gives
\begin{equation*}
    \begin{aligned}
        &g_0(x',x'',y)=P_0(x',x'',y),\quad
g'_0(x',x'',y)=P_0(x',x'',y)\oplus y_3,\\
&g_1(x',x'',y)=P_1(x',x'',y),\quad
g'_1(x',x'',y)=P_1(x',x'',y)\oplus y_3.
    \end{aligned}
\end{equation*}

Therefore, $\{g_0,g'_0,g_1,g'_1\}$ is a
maximum-cardinality family of four disjoint spectra
$2$-plateaued $\mathrm{GMM}$ functions without nonzero linear structures in $10$ variables. Moreover,
$\deg(g_0)=\deg(g'_0)=\deg(g_1)=\deg(g'_1)
=4=p+(n-m)/2$.
Thus, all four members attain the next-to-optimal algebraic degree
simultaneously.

In this example, $\gamma_0=\gamma_1$, so the sufficient condition
given in the Walsh-support comparison following Theorem~3 for
excluding common EA-equivalence is not satisfied. Nevertheless, a
direct calculation of the Walsh-support self-sums gives
\[
\begin{aligned}
\bigl|
\operatorname{supp}(W_{g_0})
\oplus
\operatorname{supp}(W_{g_0})
\bigr|
&=
\bigl|
\operatorname{supp}(W_{g'_0})
\oplus
\operatorname{supp}(W_{g'_0})
\bigr|
=736,\\
\bigl|
\operatorname{supp}(W_{g_1})
\oplus
\operatorname{supp}(W_{g_1})
\bigr|
&=
\bigl|
\operatorname{supp}(W_{g'_1})
\oplus
\operatorname{supp}(W_{g'_1})
\bigr|
=672.
\end{aligned}
\]
Thus, the Walsh-support self-sum cardinality is not constant across
the family. By Proposition~\ref{Pro-5},
this family is not memberwise EA-equivalent to any
family obtained from \cite[Theorem~4.4]{disjoint2}.
\end{Example}
\section{An application to the construction of SAO resilient functions without nonzero linear structures }\label{sec-SAO}
In this section, we use coding theory and a common invertible
linear transformation to obtain a $t$-resilient disjoint spectra
plateaued family without nonzero linear structures from a
suitably chosen family constructed in Theorem~1. We then
combine selected members of this family with linear functions,
following the framework of~\cite{2009disjoint}, to construct
an SAO $t$-resilient Boolean function without nonzero linear
structures. Finally, we modify one constituent to attain
optimal algebraic degree while preserving these properties.
\subsection{Resilient subfamilies via linear transformations and coding theory}

Let $\Omega_R^{\mathcal{F}_0}$ be the family constructed in
Theorem~\ref{The-Set-withoutlinear}. We next introduce a common
invertible linear transformation of the $x$-variables and study the
resilience of the resulting functions. For
$A\in\operatorname{GL}(m,2)$, define
$f^{A}_{kr}(x)
=
f_{kr}(Ax)
=
(A^Tc_{kr})\cdot x$ and let
\begin{equation}\label{FA}
    \mathcal{F}^A_{0}
=
\{f^A_{kr}\mid 0\leq k\leq 2^q-1,\ 
0\leq r\leq 2^p-1\}.
\end{equation}
Using $\mathcal{F}_{0}^A$ in place of $\mathcal{F}_0$, while keeping
the offset parameters $R$ in \eqref{Set-R(psi)} unchanged, yields the family
$\Omega_R^{\mathcal{F}_{0}^A}$, whose members are
\begin{equation*}
   g_{k,A}(x,y)=g_k(Ax,y),
\qquad
g_{k,A}'(x,y)=g_k'(Ax,y), 
\end{equation*}
where $(x,y)\in\F_2^m\times\F_2^{p+1}$.
The corresponding coefficient sets are
\[
\mathcal{D}_k^A
=
A^T\mathcal{D}_k
=
\{A^Tc_{kr}\mid 0\leq r\leq 2^p-1\}.
\]
By Proposition~\ref{Pro-nonconstant}(1), the sets
$\mathcal{D}_0,\ldots,\mathcal{D}_{2^q-1}$ form a partition of
$\mathbb{F}_2^m$. Since $A^T$ is a bijection of
$\mathbb{F}_2^m$, their images
$\mathcal{D}_0^A,\ldots,\mathcal{D}_{2^q-1}^A$ also form a
partition of $\mathbb{F}_2^m$.
\begin{Proposition}\label{prop-resilient-subfamily}
Let $A\in\operatorname{GL}(m,2)$ and let $1\leq t<m$. Suppose that
\begin{equation}\label{B0t}
    \mathbb{B}_m(t)=\{c\in\mathbb{F}_2^m\mid \operatorname{wt}(c)\leq t\}
\subseteq\mathcal{D}^A_{0}.
\end{equation}
Then
$\Omega_R^{(\mathcal{F}_{0}^A)}
\setminus\{g_{0,A},g_{0,A}'\}$
is a family of $2^{q+1}-2$ $t$-resilient disjoint spectra
$(q+1)$-plateaued Boolean functions without nonzero linear structures in $m+p+1$ variables, where $\mathcal{F}^A_0$ is defined in \eqref{FA} and $R$ is defined in \eqref{Set-R(psi)}.
\end{Proposition}
\begin{proof}
Since
$\mathcal{D}_0^A,\ldots,\mathcal{D}_{2^q-1}^A$ form a partition of
$\mathbb{F}_2^m$, the assumption implies that, for every
$1\leq k\leq 2^q-1$ and every $0\leq r\leq 2^p-1$,
$\operatorname{wt}(A^Tc_{kr})\geq t+1$.
Indeed, if $\operatorname{wt}(A^Tc_{kr})\leq t$, then
$A^Tc_{kr}\in\mathcal{D}_0^A$ by assumption. On the other hand,
$A^Tc_{kr}\in\mathcal{D}_k^A$, contradicting the disjointness of
$\mathcal{D}_0^A$ and $\mathcal{D}_k^A$.

Applying the common linear transformation $x\mapsto Ax$ to the
Walsh-support description in \eqref{suppgg} gives 
\begin{equation*}
    \begin{aligned}
        &\operatorname{supp}(W_{g_{k,A}})
=
\bigcup_{r=0}^{2^p-1}
\left\{
(A^Tc_{kr},\omega'')
\;\middle|\;
\gamma_k\cdot\omega''=\psi_k(v_r)
\right\},\\&
\operatorname{supp}(W_{g_{k,A}'})
=
\bigcup_{r=0}^{2^p-1}
\left\{
(A^Tc_{kr},\omega'')
\;\middle|\;
\gamma_k\cdot\omega''=\psi_k(v_r)\oplus1
\right\}.
    \end{aligned}
\end{equation*}
Consequently, every point in either Walsh support has Hamming weight
\[
\operatorname{wt}(A^Tc_{kr},\omega'')
=
\operatorname{wt}(A^Tc_{kr})
+
\operatorname{wt}(\omega'')
\geq t+1.
\]
Thus, neither Walsh support contains a vector of Hamming weight at
most $t$, and hence $g_{k,A}$ and $g_{k,A}'$ are $t$-resilient for
every $1\leq k\leq 2^q-1$.

The family
$\Omega_R^{(\mathcal{F}_0^A)}
\setminus\{g_{0,A},g_{0,A}'\}$
contains
$2^{q+1}-2$
members. Finally, the transformation
$(x,y)\mapsto(Ax,y)$ is invertible. Therefore, it preserves
plateauedness, pairwise disjointness of the Walsh supports, and the
absence of nonzero linear structures. The remaining assertions now
follow from Theorem~\ref{The-Set-withoutlinear}.
\end{proof}

The preceding proposition shows that
$\mathbb B_m(t)\subseteq\mathcal D_0^A$ is sufficient for
$\Omega_{\mathcal R}^{\mathcal F_0^A}$ to contain a
$t$-resilient subfamily of cardinality $2^{q+1}-2$.
We next use coding theory to choose $\mathcal D_0$
(and thus $\mathcal F_0$) and $A$ jointly so that this
inclusion holds.

\begin{Definition}
A $k$-dimensional subspace $\mathcal{C}$ of $\mathbb{F}_2^n$ is
called a linear code. The minimum distance of
$\mathcal{C}$ is defined by
$\delta(\mathcal{C})
=
\min_{c\in\mathcal{C}\setminus\{\mathbf{0}_n\}}
\operatorname{wt}(c)$.
If $\delta(\mathcal{C})=\delta$, then $\mathcal{C}$ is a
$[n,k,\delta]$ linear code.
\end{Definition}

\begin{Proposition}\label{code}
Let $p\geq2$,
$1\leq q<2^p-p-1$,
$m=p+q$,
and let $1\leq t<m$. Suppose that there exists a binary $[m,q,\delta]$ linear code
$\mathcal{C}$ with $\delta\geq 2t+1$. Then the coefficient set $\mathcal D_0$ in the first
construction and a matrix $A\in\operatorname{GL}(m,2)$
can be chosen such that
$\mathbb B_m(t)\subseteq\mathcal D_0^A$,
where $\mathbb B_m(t)$ is defined in \eqref{B0t}.
\end{Proposition}
\begin{proof}
Let $\mathcal{C}$ be a binary linear $[m,q,\delta]$ code with $\delta\geq2t+1$. We first prove that each coset of
$\mathcal{C}$ contains at most one vector of $\mathbb{B}_m(t)$. Indeed, if
distinct $u,u'\in \mathbb{B}_m(t)$ belonged to the same coset, then
$u\oplus u'$ would be a nonzero codeword in $\mathcal{C}$ satisfying
$\operatorname{wt}(u\oplus u')
\leq
\operatorname{wt}(u)+\operatorname{wt}(u')
\leq2t$,
contrary to the minimum-distance assumption.

Consequently, the vectors in $\mathbb{B}_m(t)$ can be extended to a complete
set $\mathcal{S}$ of coset representatives of $\mathcal{C}$. Since
$\dim(\mathcal{C})=q$ and $m=p+q$, the number of cosets is $2^p$,
and hence
$|\mathcal{S}|=2^p$ and
$\mathbb{B}_m(t)\subseteq\mathcal{S}$.

Let
$V=\{\mathbf{0}_p\}\times\mathbb{F}_2^q
\subseteq\mathbb{F}_2^m$.
Both $V$ and $\mathcal{C}$ have dimension $q$, so there exists
$A\in\operatorname{GL}(m,2)$ such that
$A^TV=\mathcal{C}$.
Define
$\mathcal{D}_0=(A^T)^{-1}\mathcal{S}$.
Since $\mathcal{S}$ is a complete set of coset representatives of
$\mathcal{C}$, the set $\mathcal{D}_0$ is a complete set of coset
representatives of $V$. The cosets of $V$ are indexed by their first
$p$ coordinates. Therefore, there is a mapping
$h:\mathbb{F}_2^p\rightarrow\mathbb{F}_2^q$
such that
\[
\mathcal{D}_0
=
\{(v_r,h(v_r))\,|\,0\leq r\leq 2^p-1\},
\]
where  $\{v_r\}_{r=0}^{2^p-1} =\mathbb{F}_2^p$.
Writing
\[
h(v_r)
=
\bigl(\phi_0(v_r),\ldots,\phi_{q-1}(v_r)\bigr)
\]
gives the Boolean functions used to define the constituent family
$\mathcal{F}_0$ in Theorem~\ref{The-Set-withoutlinear}.

It remains to verify the required linear-independence condition on
these functions. 
Since $t\geq 1$, the Hamming ball $\mathbb{B}_m(t)$ contains
$\mathbf{0}_m$ and all standard basis vectors
$e_0,\ldots,e_{m-1}$ of $\mathbb{F}_2^m$. Since
$\mathbb{B}_m(t)\subseteq \mathcal{S}$, we have
$\{\mathbf{0}_m,e_0,\ldots,e_{m-1}\}\subseteq \mathcal{S}$.
We claim that $\mathcal{S}$ has the no-hyperplane property. Indeed, suppose that
$\mathcal{S}\subseteq
\{x\in\mathbb{F}_2^m:a\cdot x=b\}$
for some $a\in\mathbb{F}_2^m\setminus\{\mathbf{0}_m\}$ and
$b\in\mathbb{F}_2$. Since $\mathbf{0}_m\in\mathcal{S}$, we obtain
$b=0$. Moreover, since $e_i\in\mathcal{S}$ for every
$0\leq i\leq m-1$, we have
$a_i=a\cdot e_i=0$.
Hence $a=\mathbf{0}_m$, a contradiction. Therefore,
$\mathcal{S}$ has the no-hyperplane property.
Since $A^T$ is invertible,
$\mathcal{D}_0=(A^T)^{-1}\mathcal{S}$ also has the no-hyperplane
property. Hence
\[
1,y_0,\ldots,y_{p-1},\phi_0,\ldots,\phi_{q-1}
\]
are linearly independent. Otherwise, there would exist
$\alpha=(\alpha_0,\alpha_1)
\in\mathbb{F}_2^p\times\mathbb{F}_2^q
\setminus\{\mathbf{0}_m\}$
and $c\in\mathbb{F}_2$ such that
$\alpha_0\cdot y'
\oplus
\bigoplus_{i=0}^{q-1}\alpha_{1i}\phi_i(y')
\oplus c
=0$
for every $y'\in\mathbb{F}_2^p$. Hence every point
$(y',h(y'))\in\mathcal{D}_0$ would lie in a proper affine
hyperplane
contradicting the no-hyperplane property of $\mathcal{D}_0$.

Thus,
$\mathcal{F}_0$ satisfies the hypotheses of
Theorem~\ref{The-Set-withoutlinear}.
Finally,
$\mathcal{D}_0^A
=
A^T\mathcal{D}_0
=
\mathcal{S}$,
and therefore
$\mathbb{B}_m(t)\subseteq\mathcal{D}_0^A$.
\end{proof}

The coding-theoretic hypothesis in the preceding proposition can be
fulfilled by three types of binary linear codes. For the BCH and Goppa
constructions, set
$\ell=\left\lceil\log_2(m+1)\right\rceil$.
Both constructions yield a binary linear $[m,k,\delta]$ code satisfying
$k\ge m-\ell t$ and $\delta\ge 2t+1$.
Therefore, since $m=p+q$, the condition
$p\ge t\left\lceil\log_2(m+1)\right\rceil$
ensures that $k\ge q$.
In this case, any $q$-dimensional subcode may be selected, and its
minimum distance is still at least $2t+1$.
We now describe the three constructions separately.

\paragraph{Primitive binary BCH codes.}
Peterson's construction gives a primitive binary BCH code of length
$2^\ell-1$, minimum distance at least $2t+1$, and codimension at
most $\ell t$~\cite{ma2}. Shortening this code to length $m$
preserves the minimum-distance lower bound and yields the parameters
stated above.

\paragraph{Square-free binary Goppa codes.}
Alternatively, one may use a square-free binary Goppa code of length
$m$, defined over $\mathbb F_{2^\ell}$ by a Goppa polynomial of
degree $t$. Such a code satisfies the same dimension and
minimum-distance bounds~\cite{ma1}.

\paragraph{The binary Golay code.}
For certain exceptional parameter choices, the required code can be
obtained from the binary Golay code. The binary Golay code has
parameters $[23,12,7]$ and can be obtained by puncturing the extended
binary Golay $[24,12,8]$ code~\cite{Golay1,Golay2}.
Hence, for $m=23$, $q\le 12$, and $t\le 3$, one may take any
$q$-dimensional subcode of the $[23,12,7]$ Golay code; its minimum
distance is at least $7\ge 2t+1$.
In particular, the full $[23,12,7]$ Golay code is used in the
$70$-variable $3$-resilient construction presented in the following
subsection.

Additional codes for specific parameter choices can be identified
from  online tables of best-known linear codes
\cite{GrasslCodeTables}, whenever the listed minimum distance is
at least $2t+1$.
\begin{Remark}\label{Remark4}
For \(p\geq2\), \(1\leq q<2^p-p-1\), \(m=p+q\), and
\(1\leq t<m\), if a binary \([m,q,\delta]\) linear code with
\(\delta\geq2t+1\) exists, 
Propositions  \ref{prop-resilient-subfamily} and \ref{code} together yield
\(2^{q+1}-2\) \(t\)-resilient disjoint spectra
\((q+1)\)-plateaued Boolean functions without nonzero linear
structures in \(m+p+1\) variables, thereby providing a
constructive lower bound on the maximum attainable cardinality.
On the other hand, each \((q+1)\)-plateaued Boolean function in \(m+p+1\)
variables has \(2^{2p}\) nonzero Walsh coefficients. Since the
Walsh support of a \(t\)-resilient function contains no vector
of Hamming weight at most \(t\), the cardinality of any
\(t\)-resilient disjoint spectra family of such functions is
bounded above by
\[
\left\lfloor
\frac{
2^{m+p+1}-\sum_{i=0}^{t}\binom{m+p+1}{i}
}{2^{2p}}
\right\rfloor
=
2^{q+1}-
\left\lceil
\frac{\sum_{i=0}^{t}\binom{m+p+1}{i}}{2^{2p}}
\right\rceil.
\]
This is a counting upper bound and need not be attainable.
Since the summation includes the zero vector, this upper bound
is at most \(2^{q+1}-1\).  Consequently, under the conditions
stated above, the maximum attainable cardinality of such
families without nonzero linear structures is either
\(2^{q+1}-2\) or \(2^{q+1}-1\). The constructed family
therefore either attains the maximum or falls short of it
by one function.
\end{Remark}

\subsection{SAO resilient functions without nonzero linear structures}
We now build on the framework for constructing resilient
Boolean functions from disjoint spectra families, first
introduced in 2009~\cite{2009disjoint}.
Selected members of the $t$-resilient family obtained in the
preceding subsection are combined with linear functions to
construct an SAO $t$-resilient Boolean function without
nonzero linear structures. We then modify one
constituent to attain optimal algebraic degree while
preserving these properties.

Let
$p\geq2$, $1\leq q<2^p-p-1$, $m=p+q$, $N=m+p+1$, and $1\leq t<m$. Define
$\Gamma_t
=\{u\in\mathbb F_2^N\,|\,\operatorname{wt}(u)\leq t\}.$
Suppose that a binary $[m,q,\delta]$ linear code with
$\delta\geq2t+1$ exists. Assume in addition that
$|\Gamma_t|=\sum_{i=0}^{t}\binom Ni\leq2^{q+1}-2.$
 By Proposition \ref{code}, the coefficient set
$\mathcal D_0$ in the first construction and a matrix
$A\in\operatorname{GL}(m,2)$ can be chosen so that
$\mathbb B_m(t)\subseteq\mathcal D_0^A$.
Proposition~\ref{prop-resilient-subfamily} then shows that
$\Omega_R^{\mathcal F_0^A}
\setminus\{g_{0,A},g'_{0,A}\}$
is a family of $2^{q+1}-2$ $t$-resilient disjoint spectra
$(q+1)$-plateaued Boolean functions without nonzero linear
structures.

We may therefore choose a subfamily
\begin{equation}\label{Ht}
    \mathcal H_t
 \subseteq
 \Omega_R^{(\mathcal F_0^A)}
 \setminus\{g_{0,A},g'_{0,A}\}
\end{equation}
of cardinality
$|\mathcal H_t|=|\Gamma_t|
 =\sum_{i=0}^{t}\binom Ni$,
and index its members as
$\mathcal H_t=\{H_u\mid u\in\Gamma_t\}$.
For every $u\notin\Gamma_t$, we have
$\operatorname{wt}(u)\geq t+1$, and hence the linear function
$u\cdot X$ is $t$-resilient. For each
$u\in\mathbb{F}_{2}^{N}$, define
\begin{equation}\label{FU}
    F_u(X)
=
\begin{cases}
H_u(X), & u\in\Gamma_{t},\\
u\cdot X, & u\notin\Gamma_{t},
\end{cases}
\end{equation}
where $X\in\F_2^{N}$. Hence every  \(F_u\) is \(t\)-resilient.
We then define
\begin{equation}\label{resilient}
    F(X,Y)=F_Y(X),
\qquad
(X,Y)\in\mathbb{F}_2^{N}\times\mathbb{F}_2^{N}.
\end{equation}

The following table summarizes the steps leading to the
construction of $F$.
\begin{table}[H]\small
\caption{Roadmap for constructing the SAO resilient function \(F\) without nonzero linear structures.}
\label{tab-resilient-roadmap}
\centering
\renewcommand{\arraystretch}{1.15}
\begin{tabularx}{\textwidth}{@{}p{1.5cm}X@{}}
\toprule
\textbf{Step 1.}
&
Choose integers $p\geq2$, $1\leq q<2^p-p-1$, and
$1\leq t<m$, and set $m=p+q$ and $N=m+p+1$.
Verify that
$\sum_{i=0}^{t}\binom{N}{i}\leq 2^{q+1}-2$.
Choose a binary $[m,q,\delta]$ linear code $\mathcal C$
with $\delta\geq 2t+1$.
\\

\textbf{Step 2.}
&
Extend
$\mathbb{B}_m(t)$
to a complete set $\mathcal S$ of coset representatives of
$\mathcal C$. Choose $A\in\operatorname{GL}(m,2)$ such that
$A^T\bigl(\{\mathbf 0_p\}\times\mathbb F_2^q\bigr)=\mathcal C$.
\\

\textbf{Step 3.}
&
Set $\mathcal D_0=(A^T)^{-1}\mathcal S$ and write
$\mathcal D_0
 =\{(v_r,h(v_r))\,|\,v_r\in\mathbb F_2^p,0\leq r\leq 2^p-1\}$,
 and
$h=(\phi_0,\ldots,\phi_{q-1})$.
For each $k\in\mathbb F_2^q$, define
$\mathcal D_k
 =\mathcal D_0\oplus\{(\mathbf 0_p,k)\}
 =\{c_{kr}\mid 0\leq r\leq2^p-1\}$,
and set $f_{kr}(x)=c_{kr}\cdot x$.
\\

\textbf{Step 4.}
&
Starting from $\mathcal F_0$, carry out Steps~4--6 of
Table~\ref{tab:construction-roadmap} to construct $g_k$ and $g'_k$ and
form $\Omega_R^{(\mathcal F_0)}$.\\
\textbf{Step 5.}
&
For $0\leq k\leq 2^q-1$, define
$g_{k,A}(x,y)=g_k(Ax,y)$,
$g'_{k,A}(x,y)=g'_k(Ax,y)$,
and form the transformed family
$\Omega_R^{(\mathcal F^A_{0})}$.
\\

\textbf{Step 6.}
&
From
$\Omega_R^{(\mathcal F^A_{0})}
 \setminus\{g_{0,A},g'_{0,A}\}$,
choose
$\sum_{i=0}^{t}\binom{N}{i}$
members and index them as
$\mathcal H_t=\{H_u\,|\, u\in\Gamma_t\}$.
\\

\textbf{Step 7.}
&
For each $u\in\mathbb F_2^N$, define \(F_u(X)\) as in \eqref{FU}
and define $F(X,Y)=F_Y(X)$.
\\
\bottomrule
\end{tabularx}
\end{table}
\begin{Theorem}\label{FFF}
The Boolean function \(F\) defined in \eqref{resilient} is a \(2N\)-variable SAO
\(t\)-resilient Boolean function without nonzero linear structures and with
nonlinearity
$N_F=2^{2N-1}-2^{N-1}-2^m$.
\end{Theorem}
\begin{proof} 
 We first prove that $F$ admits no nonzero linear structure.  Suppose that
$(\alpha,\beta)\in\mathbb{F}_2^{N}\times\mathbb{F}_2^{N}$ is a linear structure of
$F$. 
If $\beta=\boldsymbol{0}_N$, then, for any $u\in\Gamma_t$, the restriction of
$D_{(\alpha,\boldsymbol{0}_N)}F$ to $Y=u$ is $D_{\alpha}H_u$. Since $H_u$ admits no nonzero
linear structure, it follows that $\alpha=\boldsymbol{0}_N$.
Now suppose that \(\beta\neq\boldsymbol{0}_N\). Since \(N-1=2p+q\) and
\(p\geq2\), we have \(|\Gamma_t|\leq 2^{q+1}-2<2^{N-1}\), and thus
\(|\Gamma_t\cup(\Gamma_t\oplus\beta)|\leq2|\Gamma_t|<2^N\).
Therefore, \(\Gamma_t\cup(\Gamma_t\oplus\beta)\) is a proper subset of
\(\mathbb F_2^N\), so one can choose
\(u\notin\Gamma_t\cup(\Gamma_t\oplus\beta)\). Then we have
\[
\begin{aligned}
D_{(\alpha,\beta)}F(X,u)
&=F(X\oplus \alpha,u\oplus \beta)\oplus F(X,u)\\
&=(u\oplus \beta)\cdot(X\oplus \alpha)\oplus u\cdot X\\
&=\beta\cdot X\oplus(u\oplus \beta)\cdot \alpha,
\end{aligned}
\]
which is nonconstant in $X$, a contradiction. Thus $(\alpha,\beta)=(\boldsymbol{0}_N,\boldsymbol{0}_N)$,
and consequently $F$ admits no nonzero linear structure.

It remains to show that \(F\) is \(t\)-resilient and to determine its nonlinearity. For
$(\omega',\omega'')\in\mathbb F_2^N\times\mathbb F_2^N$, the Walsh
transform of $F$ is
\[
W_F(\omega',\omega'')
 =\sum_{u\in\mathbb F_2^N}
  (-1)^{\omega''\cdot u}W_{F_u}(\omega')=\sum_{u\in\Gamma_t}(-1)^{\omega''\cdot u}W_{H_u}(\omega')
  +\sum_{u\notin\Gamma_t}(-1)^{\omega''\cdot u}
       W_{u\cdot X}(\omega').
\]
If \(\operatorname{wt}(\omega',\omega'')\leq t\), then
\(\operatorname{wt}(\omega')\leq t\). Since every \(F_u\) is
\(t\)-resilient, \(W_{F_u}(\omega')=0\) for every \(u\).
Therefore \(W_F(\omega',\omega'')=0\), proving that \(F\)
is \(t\)-resilient.

By the pairwise disjointness of the Walsh supports of the
members of $\mathcal H_t$, at most one term in the first sum
is nonzero for each $\omega'$, and its absolute value is
$2^{m+1}$. The second sum is either zero or has absolute
value $2^N$. Hence
$|W_F(\omega',\omega'')|\leq 2^N+2^{m+1}$
for all $(\omega',\omega'')\in
\mathbb F_2^N\times\mathbb F_2^N$.
To see that equality holds, choose $u_0\in\Gamma_t$ and
$\omega'_0\in\operatorname{supp}(W_{H_{u_0}})$. Since $H_{u_0}$ is
$t$-resilient, $\omega'_0\notin\Gamma_t$, and thus
$\omega'_0\neq u_0$. Then
 we have
\begin{equation*}
\begin{aligned}
        W_F(\omega_0',\omega'')
 =&(-1)^{\omega''\cdot u_0}W_{H_{u_0}}(\omega_0')
  +(-1)^{\omega''\cdot \omega_0'}2^N
  \\=&(-1)^{\omega''\cdot u_0}\big(W_{H_{u_0}}(\omega_0')
  +(-1)^{\omega''\cdot (\omega_0'\oplus u_0)}2^N\big)
\end{aligned}
\end{equation*}
Since \(u_0\neq \omega_0'\), we have \(u_0\oplus \omega_0'\neq 0\).
Hence one can choose \(\omega_0''\in\mathbb F_2^N\) so that the two
nonzero terms in \(W_F(\omega_0',\omega_0'')\) have the same sign.
Therefore,
$|W_F(\omega'_0,\omega''_0)|=2^N+2^{m+1}$.
Consequently,
\begin{equation*}
    N_F
 =2^{2N-1}
  -\frac{1}{2}\max_{(\omega',\omega'')\in\F_2^N\times\F_2^N}|W_F(\omega',\omega'')|
 =2^{2N-1}-2^{N-1}-2^m.
\end{equation*}
Finally, $N=m+p+1$ and $p\geq2$ imply
$2^{N-1}+2^m<2^N$. Hence
$N_F>2^{2N-1}-2^N$, so $F$ is SAO.
\end{proof}
\begin{Proposition}\label{cishu}
Under the hypotheses of Theorem~\ref{FFF},
the function $F$ satisfies
$\deg(F)\le N+p+1$.
Suppose additionally that
the construction of $\Omega_R^{(\mathcal F^A_0)}$ uses a common
selector map; that is, the maps $\sigma_k$ defined in \eqref{smap}
satisfy
$\sigma_k=\sigma$
for every $0\leq k\leq 2^q-1$.
Suppose further that the functions $\psi_k$ associated with the
offset family $R$ in \eqref{Set-R(psi)} all coincide with a single
function $\psi\in\mathcal B_p\setminus U$ of degree $p$, where $U$ is defined in $\eqref{spanU}$.
Then the function $F$ satisfies
\[
\deg(F)=
\begin{cases}
N+p+1, & |\Gamma_t|\text{ is odd},\\
\leq N+p, & |\Gamma_t|\text{ is even}.
\end{cases}
\]
\end{Proposition}
\begin{proof}
The function $F$ has the equivalent indicator representation
$F(X,Y)
=
\bigoplus_{u\in\mathbb F_2^N}
\Big(
\prod_{i=0}^{N-1}(Y_i\oplus u_i\oplus1)
\Big)F_u(X)$,
where $Y=(Y_0,\ldots,Y_{N-1})$ and
$u=(u_0,\ldots,u_{N-1})$. For every \(u\in\mathbb F_2^N\), the unique term of degree \(N\)
in
$\prod_{i=0}^{N-1}(Y_i\oplus u_i\oplus 1)$
is \(Y_0Y_1\cdots Y_{N-1}\). Hence the coefficient of
\(Y_0Y_1\cdots Y_{N-1}\) in \(F\) is
\[
\bigoplus_{u\in\mathbb F_2^N}F_u(X)
=
\Big(\bigoplus_{u\in\Gamma_t}H_u(X)\Big)
\oplus
\Big(\bigoplus_{u\notin\Gamma_t}u\cdot X\Big).
\]
The second term has algebraic degree at most $1$. Since
$H_u\in\Omega_{ R}^{(\mathcal F_0^A)}
\setminus\{g_{0,A},g'_{0,A}\}$ for every $u\in\Gamma_t$, and
every member of $\Omega_{R}^{(\mathcal F_0^A)}$
has algebraic degree at most $p+1$, we have
$\deg(F_u)\le p+1$ for every $u\in\mathbb F_2^N$.
It then follows from the indicator representation of $F$ that
$\deg(F)\le N+p+1$.

Under the additional assumptions of this proposition, by
Corollary~\ref{Cor1} and the invariance of algebraic degree under
invertible linear transformations, every member of
$\Omega_R^{(\mathcal F_0^A)}$ has algebraic degree $p+1$. We next
show that all these members have the same homogeneous component of
degree $p+1$.

Since the construction uses the common selector map $\sigma$, for
every $0\leq k\leq2^q-1$ we have
\[
g_{k,A}(x,y)
=g_k(Ax,y)\\
=c_{k,\sigma(y)}\cdot Ax\oplus a_k(y)
=A^Tc_{k,\sigma(y)}\cdot x\oplus a_k(y).
\]
Since
$c_{k,\sigma(y)}
=
c_{0,\sigma(y)}\oplus(0_p,k)$,
it follows that
\begin{equation*}
    \begin{aligned}
       &g_{k,A}(x,y)
=
A^Tc_{0,\sigma(y)}\cdot x
\oplus A^T(0_p,k)\cdot x
\oplus a_k(y),\\
&g'_{k,A}(x,y)
=
A^Tc_{0,\sigma(y)}\cdot x
\oplus A^T(0_p,k)\cdot x
\oplus b_k(y). 
    \end{aligned}
\end{equation*}
Thus, the part involving $x$ that may have degree $p+1$ is
independent of $k$, because $A^T(0_p,k)\cdot x$ is linear.

It remains to consider the offset functions. By \eqref{Mob}, the coefficient of the monomial
$y_0\cdots y_p$ in the algebraic normal form of $a_k$ is
$\bigoplus_{y\in\mathbb F_2^{p+1}}a_k(y)
=
\bigoplus_{r=0}^{2^p-1}
   \bigl(a_{kr}\oplus a'_{kr}\bigr)
=
\bigoplus_{r=0}^{2^p-1}\psi_k(v_r)
=
1$,
where the last equality follows from $\deg(\psi_k)=p$. Likewise,
$\bigoplus_{y\in\mathbb F_2^{p+1}}b_k(y)
=
\bigoplus_{r=0}^{2^p-1}
   \bigl(b_{kr}\oplus b'_{kr}\bigr)
=
\bigoplus_{r=0}^{2^p-1}
   \bigl(\psi_k(v_r)\oplus1\bigr)
=
1.$

Therefore, for every \(0\leq k\leq 2^q-1\), the degree-\((p+1)\)
parts of \(g_{k,A}\) and \(g'_{k,A}\) coincide. They consist of the
degree-\((p+1)\) terms arising from
\(A^Tc_{0,\sigma(y)}\cdot x\), together with the monomial
\(y_0\cdots y_p\). This common degree-\((p+1)\) part is nonzero,
since every monomial arising from \(A^Tc_{0,\sigma(y)}\cdot x\)
contains at least one variable from \(x\), whereas
\(y_0\cdots y_p\) does not.

Since every \(H_u\), \(u\in\Gamma_t\), belongs to
\(\Omega_R^{(\mathcal F_0^A)}\), all the functions \(H_u\) have the
same nonzero degree-\((p+1)\) part. Consequently, this common part
survives in
$\bigoplus_{u\in\Gamma_t}H_u$
if \(|\Gamma_t|\) is odd, whereas it cancels completely if
\(|\Gamma_t|\) is even.
Then we obtain
\[
\deg
\Big(
\bigoplus_{u\in\Gamma_t}H_u
\Big)
=
\begin{cases}
p+1, & |\Gamma_t|\text{ is odd},\\
\leq p, & |\Gamma_t|\text{ is even}.
\end{cases}
\]

On the other hand, for any monomial in \(Y\) of degree less than \(N\),
its coefficient in \(F\) has algebraic degree at most \(p+1\).
Hence, the total degree of the corresponding term is at most $N+p$.

Combining the above observations, the assertion follows.
\end{proof}

The preceding result gives $\deg(F)\leq N+p+1$, whereas the  bound for a $2N$-variable $t$-resilient Boolean function
is $2N-t-1$. Moreover, $|\Gamma_t|\leq2^{q+1}-2$ forces
$t\leq m-2$, since $t=m-1$ would imply
$|\Gamma_t|
=\sum_{i=0}^{m-1}\binom{N}{i}
\geq2^{m-1}\geq2^{q+1}$,
a contradiction. Hence, as $N=m+p+1$,
$(2N-t-1)-(N+p+1)=m-t-1\geq1$.
Thus, the constructed function does not attain the optimal algebraic
degree for resilient Boolean functions.
We next modify $F$ to attain the optimal algebraic degree
while preserving its $t$-resiliency, SAO property, and absence
of nonzero linear structures.

Choose
$u'\in\mathbb F_2^N$ with
$\operatorname{wt}(u')=t+1$.
Define
\[
Q_{u'}(X)
=u'\cdot X\oplus \prod_{i\notin\operatorname{supp}(u')}X_i,
\]
where $X\in\F_2^N$.
It is easy to see that
$\deg(Q_{u'})=N-t-1$.
Replace the linear function indexed by $u'$ with
$Q_{u'}$ and define
\[
F'_u(X)
=
\begin{cases}
H_u(X), & u\in\Gamma_t,\\
Q_{u'}(X), & u=u',\\
u\cdot X, & u\notin\Gamma_t\cup\{u'\},
\end{cases}
\]
where $H_u$ denotes the member indexed by $u$ in the
subfamily $\mathcal H_t$ of \eqref{Ht}.
The modified concatenation is then defined by
\begin{equation}\label{F'}
    F'(X,Y)=F_Y'(X),
\qquad
(X,Y)\in\mathbb F_2^N\times\mathbb F_2^N.
\end{equation}

\begin{Theorem}\label{zuiyou}
   The Boolean function $F'$ defined in \eqref{F'} is a
$2N$-variable SAO $t$-resilient Boolean function without nonzero linear
structures. Moreover, $F'$ has optimal algebraic degree
$\deg(F')=2N-t-1$
and nonlinearity satisfying
$N_{F'}
\geq
2^{2N-1}-2^{N-1}-2^m-2^{t+1}$.
\end{Theorem}
\begin{proof}
We first prove that \(F'\) admits no nonzero linear structure. Suppose that
\((\alpha,\beta)\in\mathbb F_2^N\times\mathbb F_2^N\) is a linear
structure of \(F'\). If \(\beta=\boldsymbol{0}_N\), the argument used
above for \(F\) in Theorem \ref{FFF} applies unchanged and yields
\(\alpha=\boldsymbol{0}_N\).

Now suppose that $\beta\neq\boldsymbol{0}_N$. The set
$\Gamma_t\cup(\Gamma_t\oplus \beta)\cup\{u',u'\oplus \beta\}$ has
cardinality at most
$2|\Gamma_t|+2\leq2^{q+2}-2<2^N$. We can therefore choose
$u\in\mathbb F_2^N$ such that
$u,u\oplus \beta\notin\Gamma_t\cup\{u'\}$. For this $u$, both
corresponding constituent functions are linear, and
\[
\begin{aligned}
D_{(\alpha,\beta)}F'(X,u)=&
F'_u(X)\oplus F'_{u\oplus \beta}(X\oplus \alpha)
\\=&
u\cdot X\oplus(u\oplus \beta)\cdot(X\oplus \alpha)\\
=&
\beta\cdot X\oplus(u\oplus \beta)\cdot \alpha.
\end{aligned}
\]
This function is nonconstant because $\beta\neq\boldsymbol{0}_N$, a contradiction.
Thus $(\alpha,\beta)=(\boldsymbol{0}_N,\boldsymbol{0}_N)$, and $F'$ admits no nonzero linear structure.

We next show that \(Q_{u'}\) is also \(t\)-resilient. Since
\((-1)^{\prod_{i\notin\operatorname{supp}(u')}X_i}
=1-2\prod_{i\notin\operatorname{supp}(u')}X_i\), we have
\begin{align*}
W_{Q_{u'}}(\omega')
&=\sum_{X\in\mathbb F_2^N}
  (-1)^{(u'\oplus\omega')\cdot X}-2
  \sum_{X\in\mathbb F_2^N}\big(\prod_{i\notin\operatorname{supp}(u')}X_i
  (-1)^{(u'\oplus\omega')\cdot X}\big)\\=&
\sum_{X\in\mathbb F_2^N}
  (-1)^{(u'\oplus\omega')\cdot X}-2
  \sum_{\substack{X\in\mathbb F_2^N\\
  X_i=1\ \text{for all }i\notin\operatorname{supp}(u')}}
  (-1)^{(u'\oplus\omega')\cdot X}.
\end{align*}
The first sum is \(2^N\) when \(\omega'=u'\) and zero
otherwise. In the second sum, the \(t+1\) coordinates indexed
by \(\operatorname{supp}(u')\) are free. Thus, this sum vanishes
unless \(\omega'_i=1\) for every \(i\in\operatorname{supp}(u')\);
when this condition holds, it equals
\(2^{t+1}(-1)^{\sum_{i\notin\operatorname{supp}(u')}\omega'_i}\).
Consequently,
\[
W_{Q_{u'}}(\omega')=
\begin{cases}
2^N-2^{t+2},
  & \omega'=u',\\
-2^{t+2}
 (-1)^{\sum_{j\notin\operatorname{supp}(u')}\omega'_j},
  & \operatorname{supp}(u')
    \subsetneq\operatorname{supp}(\omega'),\\
0,
  & \text{otherwise},
\end{cases}
\]
which implies \(W_{Q_{u'}}(\omega')=0\) whenever
\(\operatorname{wt}(\omega')\leq t\).
Thus \(Q_{u'}\) is \(t\)-resilient. 
For
$(\omega',\omega'')\in\mathbb F_2^N\times\mathbb F_2^N$,
the Walsh transform of $F'$ is
\begin{align*}
W_{F'}(\omega',\omega'')
=&\sum_{u\in\mathbb F_2^N}
  (-1)^{\omega''\cdot u}W_{F'_u}(\omega')\\=&\sum_{u\in\Gamma_t}
  (-1)^{\omega''\cdot u}W_{H_u}(\omega')
  +(-1)^{\omega''\cdot u'}W_{Q_{u'}}(\omega')+
  \sum_{u\notin\Gamma_t\cup\{u'\}}
  (-1)^{\omega''\cdot u}W_{u\cdot X}(\omega').
\end{align*}
Since $Q_{u'}$ is $t$-resilient and $F'_u=F_u$ for
$u\ne u'$, every $F'_u$ is $t$-resilient. The same
Walsh-transform argument as in the proof of Theorem \ref{FFF}
then shows that $F'$ is $t$-resilient.

We now bound its nonlinearity. The Walsh transform of
$u'\cdot X$ equals $2^N$ at $\omega'=u'$ and zero
elsewhere. Comparing this with the formula for
$W_{Q_{u'}}$ gives
$\bigl|W_{Q_{u'}}(\omega')-W_{u'\cdot X}(\omega')\bigr|
\leq 2^{t+2}$
for every $\omega'\in\mathbb F_2^N$. Since $F'$ differs
from $F$ only in the constituent indexed by $u'$,
\[W_{F'}(\omega',\omega'')
=
W_F(\omega',\omega'')
+(-1)^{\omega''\cdot u'}
\bigl(W_{Q_{u'}}(\omega')-W_{u'\cdot X}(\omega')\bigr).\]
By Theorem~\ref{FFF},
$\max_{(\omega',\omega'')\in\F_2^N\times\F_2^N}|W_F(\omega',\omega'')|
=2^N+2^{m+1}$. Consequently,
\[
\max_{(\omega',\omega'')\in\F_2^N\times\F_2^N}|W_{F'}(\omega',\omega'')|
\leq 2^N+2^{m+1}+2^{t+2},
\]
and hence
$N_{F'}\geq
2^{2N-1}-2^{N-1}-2^m-2^{t+1}$.
Finally, $t<m$ and $p\geq2$ imply
$2^{N-1}+2^m+2^{t+1}<2^N$, so $F'$ is SAO.

Finally, we determine the algebraic degree of $F'$.
Since $\operatorname{wt}(u')=t+1$, we have
$u'\notin\Gamma_t$ and $F_{u'}(X)=u'\cdot X$.
Thus,
$F'_{u'}(X)\oplus F_{u'}(X)
=\prod_{i\notin\operatorname{supp}(u')}X_i$,
while $F'_u=F_u$ for every $u\ne u'$. It follows that
\[
F'(X,Y)=F(X,Y)\oplus
\Big(\prod_{i=0}^{N-1}(Y_i\oplus u'_i\oplus1)\Big)
\Big(\prod_{i\notin\operatorname{supp}(u')}X_i\Big).
\]
By Proposition \ref{cishu},
 $\deg(F)\leq N+p+1$. Moreover, $N+p+1<2N-t-1$, since
$(2N-t-1)-(N+p+1)=m-t-1\ge1$, as shown above.
The additional term in the expression for $F'$ has degree
$2N-t-1$, and
it cannot be cancelled by any term of $F$, 
Consequently, $\deg(F')=2N-t-1$, which attains the
algebraic-degree bound for a $2N$-variable $t$-resilient
Boolean function.
\end{proof}
\begin{Example}\label{resilientexm}
Let $(p,q,m,N)=(4,3,7,12)$ and $t=1$.
The parameter condition in Theorem~\ref{The-Set-withoutlinear} is
satisfied since
$3<2^4-4-1=11$. Let $\mathcal C$ be the binary $[7,3,4]$ simplex
code generated by
\[
G_{\mathcal{C}}=
\begin{pmatrix}
1&0&0&1&1&0&1\\
0&1&0&1&0&1&1\\
0&0&1&0&1&1&1
\end{pmatrix}.
\]
Thus, the minimum distance of $\mathcal C$ is $4\geq2t+1$. Then
$\mathbb{B}_7(1)=\{0,1,2,4,8,16,32,64\}$. Choose
\[
\mathcal{S}=\{0,1,2,3,4,5,6,8,9,10,11,12,15,16,32,64\}.
\]
Hence $\mathcal{S}$ is a complete set of coset representatives of
$\mathcal C$ and satisfies $\mathbb{B}_7(1)\subseteq S$. 

Choose $A\in\operatorname{GL}(7,2)$ such that
\begin{equation*}
    A^T=
\begin{pmatrix}
1&0&0&0&1&0&0\\
0&1&0&0&0&1&0\\
0&0&1&0&0&0&1\\
0&0&0&1&1&1&0\\
0&0&0&0&1&0&1\\
0&0&0&0&0&1&1\\
0&0&0&0&1&1&1
\end{pmatrix}.
\end{equation*}
Hence
$A^T\bigl(\{0_4\}\times\mathbb F_2^3\bigr)=\mathcal C$.
 Direct computation
gives
\[
\begin{aligned}
\mathcal{D}_0=(A^T)^{-1}\mathcal{S}
=\{&
(0,0),(1,0),(2,0),(3,0),
(4,0),(5,0),(6,0),(7,7),
(8,0),\\&(9,0),(10,0),(11,0),
(12,0),(13,5),(14,6),(15,0)
\}.
\end{aligned}
\]
Thus,
$\mathcal{D}_0=\{(y',h(y'))\mid y'\in\mathbb F_2^4\}$,
where
$h=(\phi_0,\phi_1,\phi_2):\mathbb F_2^4\to\mathbb F_2^3$ is given
by
$\phi_0(y')
=y_0y_2y_3\oplus y_0y_1y_2,
\phi_1(y')
=y_1y_2y_3\oplus y_0y_1y_2$ and
$\phi_2(y')
= y_0y_1y_2y_3\oplus y_1y_2y_3\oplus y_0y_2y_3
  \oplus y_0y_1y_2$,
for $y'=(y_0,y_1,y_2,y_3)\in\mathbb F_2^4$.
Use this function $h$ to construct $\mathcal F_0$ as in
Theorem~\ref{The-Set-withoutlinear}.

For every $0\leq k\leq7$, choose
$S_k=\mathbb F_2^4\times\{0\}$,
$\gamma_k=(\boldsymbol{0}_4,1)$,
and
$\mathcal L_k(y')=(y',0)$ for $y'\in\F_2^4$.
Thus, all the selector maps $\sigma_k$ coincide. Choose
$\psi_k(y_0,y_1,y_2,y_3)=\psi(y_0,y_1,y_2,y_3)=y_0y_1y_2y_3$
for every $0\leq k\leq7$.
For the function $h$ above, a direct verification gives
$\psi\notin U$, and clearly $\deg(\psi)=4$. Choose the offset
family $R$ by taking
$a_{kr}=b_{kr}=0$,
$a'_{kr}=\psi_k(v_r)$,
$b'_{kr}=\psi_k(v_r)\oplus1$.
The preceding construction then produces the family
$\Omega_R^{(\mathcal F_0^A)}$.

Since
$|\Gamma_1|
=
13
<
2^{q+1}-2
=
14$,
we take
$\mathcal H_1
=
\Omega_R^{(\mathcal F_0^A)}
\setminus
\{g_{0,A},g'_{0,A},g'_{7,A}\}$
and index its $13$ members by the vectors in
$\Gamma_1=\{0,1,2,4,8,16,32,64,128,256,512,1024,2048\}$.
Define $F$ by the concatenation framework above. Then $F$ is a
$24$-variable $1$-resilient Boolean function without nonzero linear
structures. Since $|\Gamma_1|=13$ is odd, Proposition \ref{cishu}
gives
$\deg(F)=N+p+1=17$.
Its nonlinearity is
$N_F
=
2^{23}-2^{11}-2^7$.

We now apply the algebraic-degree modification to the same function.
Choose
$u'=3\in\mathbb F_2^{12}$
and define
$Q_{u'}(X)
=
X_0\oplus X_1\oplus \prod_{i=2}^{11}X_i$.
Replacing only the linear
constituent indexed by $u'$ with $Q_{u'}$, as in \eqref{F'}, gives
the modified concatenation $F'$. By
Theorem~\ref{zuiyou}, $F'$ is a $24$-variable
$1$-resilient Boolean function without nonzero linear structures and
$\deg(F')=2N-t-1=22$,
which is optimal.
A direct computation gives
$N_{F'}=2^{23}-2^{11}-2^7-2^2$.
\end{Example}
Table~\ref{tab:comparison-24} compares the functions in Example~\ref{resilientexm} with existing $24$-variable SAO $1$-resilient Boolean functions.
\begin{table}[H]
\caption{Comparison of $24$-variable SAO $1$-resilient Boolean
functions.}
\label{tab:comparison-24}
\centering
\small
\setlength{\tabcolsep}{4pt}
\renewcommand{\arraystretch}{1.12}
\begin{tabular}{@{}lcccc@{}}
\toprule
\makecell[l]{Construction}
& \makecell[c]{Nonlinearity}
& \makecell[c]{Algebraic\\degree}
& \makecell[c]{Optimal algebraic\\degree attained}
& \makecell[c]{No nonzero\\linear structures}
\\
\midrule
\multicolumn{1}{c}{$F$ }
& $2^{23}-2^{11}-2^7$
& $17$
& No
& Yes
\\
\multicolumn{1}{c}{$F'$ }
& $2^{23}-2^{11}-2^7-2^2$
& $22$
& Yes
& Yes
\\
\addlinespace[2pt]
\multicolumn{1}{c}{\cite{2009disjoint}}
& $2^{23}-2^{11}-2^8-2^2$
& $22$
& Yes
& Not established
\\
\multicolumn{1}{c}{\cite{resi1}}
& $2^{23}-2^{11}-2^7-2^2$
& $22$
& Yes
& Not established
\\
\multicolumn{1}{c}{\cite{resi2}}
& $2^{23}-2^{11}-2^6$
& $18$
& No
& Yes
\\
\bottomrule
\end{tabular}
\end{table}
Compared with the construction of Li et al.~\cite{resi2},
the modified function $F'$ attains the optimal algebraic degree
$2N-2$, whereas the former has maximum attainable algebraic degree $3N/2$.  Thus, its
distance from the optimal degree is at least $N/2-2$ and
increases with $N$. Both constructions guarantee the absence
of nonzero linear structures. This degree advantage of $F'$
comes at the cost of nonlinearity: for $2N=24$, the nonlinearity
of the construction of Li et al. exceeds that of $F'$ by $68$.
More generally, subtracting our guaranteed lower bound
$2^{2N-1}-2^{N-1}-2^m-2^2$ from their nonlinearity
$2^{2N-1}-2^{N-1}-2^{N/2}$ gives
$2^m+2^2-2^{N/2}$. Nevertheless, their construction is restricted to
$1$-resilient functions. We next demonstrate the higher-order applicability of our method by
constructing $2$- and $3$-resilient SAO Boolean functions with optimal
algebraic degree and no nonzero linear structures.

More precisely, take $p=11$ and $t=2$ and use the binary
$[19,8,5]$ code specified in Appendix. Then $N=31$,
$5=2t+1$, and
$\sum_{i=0}^{2}\binom{31}{i}=497<2^9-2=510$.
The construction therefore yields a $62$-variable SAO
$2$-resilient Boolean function with optimal algebraic degree
$59$ and nonlinearity at least
$2^{61}-2^{30}-2^{19}-2^3$.
For the binary $[23,12,7]$ Golay code, take $p=11$ and $t=3$.
In this case, $N=35$, $7=2t+1$, and
$\sum_{i=0}^{3}\binom{35}{i}=7176<2^{13}-2=8190$.
We thus obtain a $70$-variable SAO $3$-resilient Boolean
function with optimal algebraic degree $66$ and nonlinearity
at least $2^{69}-2^{34}-2^{23}-2^4$.
Neither function admits a nonzero linear structure.
Table~\ref{tab:comparison-62} compares our $62$-variable
$2$-resilient function with previously known constructions.
\begin{table}[H]\small
\caption{Comparison of $62$-variable SAO $2$-resilient
Boolean functions.}
\label{tab:comparison-62}
\centering
\small
\renewcommand{\arraystretch}{1.15}
\begin{tabular}{@{}ccccc@{}}
\toprule
\makecell[c]{Construction}
&
\makecell[c]{Nonlinearity}
&
\makecell[c]{Algebraic\\degree}
&
\makecell[c]{Optimal algebraic\\degree attained}
&
\makecell[c]{No nonzero\\linear structures}
\\
\midrule
This work
&
$\geq2^{61}-2^{30}-2^{19}-2^3$
&
$59$
&
Yes
&
Yes
\\
\cite{2009disjoint}
&
$2^{61}-2^{30}-2^{19}-2^{18}-2^3$
&
$59$
&
Yes
&
Not established
\\
\cite{resi1}
&
$2^{61}-2^{30}-2^{19}-2^3$
&
$59$
&
Yes
&
Not established
\\
\bottomrule
\end{tabular}
\end{table}

\section{Conclusion}\label{con}
We have developed two explicit algebraic constructions of maximum-cardinality families of disjoint spectra plateaued Boolean functions without nonzero linear structures. Under the conditions established in Sections \ref{sec-plateauedwithoutlinear} and \ref{sec-partially-bent}, both constructions allow a prescribed common algebraic degree: all members of the first family can simultaneously attain the optimal degree, while those of the second can simultaneously attain the next-to-optimal degree. We give conditions excluding common EA-equivalence with the spectral construction of \cite[Theorem~4.4]{disjoint2}, as well as examples from both constructions that are not even memberwise EA-equivalent to any family obtained from it. As an application of the first construction, the
coding-theoretic condition in Section~\ref{sec-SAO} yields $t$-resilient
disjoint spectra plateaued subfamilies without nonzero linear
structures that are either maximum-cardinality or fall short
by only one member. Under the additional cardinality condition,
selected members serve as constituents in the framework
of~\cite{2009disjoint}, producing SAO $t$-resilient Boolean
functions without nonzero linear structures. Furthermore, modifying one linear constituent yields functions
with optimal algebraic degree while preserving SAO
$t$-resiliency and the absence of nonzero linear structures.
Instances are given for $t=1,2,3$.

Two problems remain open:
\begin{enumerate}
\item Does there exist a maximum-cardinality family of disjoint
spectra plateaued functions without nonzero linear structures
whose members all lie outside \(\mathcal{GMM}\)?
\item   For arbitrary admissible numbers of variables,
plateaued orders, and resiliency orders, determine the
maximum attainable cardinality of families of resilient
disjoint spectra plateaued Boolean functions without
nonzero linear structures, and give explicit constructions
attaining this maximum.
\end{enumerate}
\appendix[A Binary {$[19,8,5]$} Code for the $62$-Variable Construction]
\label{app:62-variable}
 $\mathcal C$ is the binary linear code generated by
$G_{\mathcal C}=[\,G_{16} \mid \boldsymbol{0}_{8\times3}\,]$,
where $\boldsymbol{0}_{8\times3}$ is the $8\times3$ zero matrix and
$G_{16}$ is given below.
\begin{equation*}\small
    G_{16}=
\begin{pmatrix}
1&0&0&1&1&1&0&0&1&0&0&0&0&0&0&0\\
0&1&0&0&1&1&1&0&0&1&0&0&0&0&0&0\\
0&0&1&0&0&1&1&1&0&0&1&0&0&0&0&0\\
0&0&0&1&0&0&1&1&1&0&0&1&0&0&0&0\\
0&0&0&0&1&0&0&1&1&1&0&0&1&0&0&0\\
0&0&0&0&0&1&0&0&1&1&1&0&0&1&0&0\\
0&0&0&0&0&0&1&0&0&1&1&1&0&0&1&0\\
0&0&0&0&0&0&0&1&0&0&1&1&1&0&0&1
\end{pmatrix}.
\end{equation*}

	\end{document}